\documentclass[12pt,american]{article}
\usepackage[T1]{fontenc}
\usepackage{geometry}
\usepackage{color}
\usepackage{babel}
\usepackage{array}
\usepackage{mathrsfs}
\usepackage{enumitem}
\usepackage{multirow}
\usepackage{amsmath}
\usepackage{amsthm}
\usepackage{amssymb}
\usepackage{stmaryrd}
\usepackage{graphicx}
\usepackage{setspace}
\usepackage[authoryear]{natbib}
\usepackage[unicode=true,pdfusetitle,
 bookmarks=true,bookmarksnumbered=false,bookmarksopen=false,
 breaklinks=false,pdfborder={0 0 1},backref=false,colorlinks=false]
 {hyperref}
\hypersetup{
 colorlinks,linkcolor={red!70!black},citecolor={blue!50!black},urlcolor={blue!60!black}}

\makeatletter

\providecommand{\tabularnewline}{\\}

\theoremstyle{plain}
\newtheorem{assumption}{\protect\assumptionname}
\theoremstyle{plain}
\newtheorem*{assumption*}{\protect\assumptionname}
\theoremstyle{remark}
\newtheorem{rem}{\protect\remarkname}
\theoremstyle{plain}
\newtheorem{lem}{\protect\lemmaname}
\theoremstyle{plain}
\newtheorem{thm}{\protect\theoremname}
\theoremstyle{remark}
\newtheorem*{rem*}{\protect\remarkname}

\usepackage{lmodern}
\usepackage{babel}
\usepackage{float}
\usepackage{mathrsfs}
\usepackage{breakurl}
\usepackage{bbm}
\usepackage{tikz}

\makeatletter
\newcommand{\customlabel}[2]{%
   \protected@write \@auxout {}{\string \newlabel {#1}{{#2}{\thepage}{#2}{#1}{}} }%
   \hypertarget{#1}{}
}
\makeatother
\allowdisplaybreaks

\makeatother

\providecommand{\assumptionname}{Assumption}
\providecommand{\lemmaname}{Lemma}
\providecommand{\remarkname}{Remark}
\providecommand{\theoremname}{Theorem}

\begin{document}
\global\long\def\a{\alpha}%
 
\global\long\def\b{\beta}%
 
\global\long\def\g{\gamma}%
 
\global\long\def\d{\delta}%
 
\global\long\def\e{\epsilon}%
 
\global\long\def\l{\lambda}%
 
\global\long\def\t{\theta}%
 
\global\long\def\o{\omega}%
 
\global\long\def\s{\sigma}%

\global\long\def\G{\Gamma}%
 
\global\long\def\D{\Delta}%
 
\global\long\def\L{\Lambda}%
 
\global\long\def\T{\Theta}%
 
\global\long\def\O{\Omega}%
 
\global\long\def\R{\mathbb{R}}%
 
\global\long\def\N{\mathbb{N}}%
 
\global\long\def\Q{\mathbb{Q}}%
 
\global\long\def\I{\mathbb{I}}%
 
\global\long\def\P{\mathbb{P}}%
 
\global\long\def\E{\mathbb{E}}%
\global\long\def\B{\mathbb{\mathbb{B}}}%
\global\long\def\S{\mathbb{\mathbb{S}}}%

\global\long\def\X{{\bf X}}%
\global\long\def\cX{\mathscr{X}}%
 
\global\long\def\cY{\mathscr{Y}}%
 
\global\long\def\cA{\mathscr{A}}%
 
\global\long\def\cB{\mathscr{B}}%
 
\global\long\def\cM{\mathscr{M}}%
\global\long\def\cN{\mathcal{N}}%
\global\long\def\cG{\mathcal{G}}%
\global\long\def\cC{\mathcal{C}}%
\global\long\def\sp{\,}%

\global\long\def\es{\emptyset}%
 
\global\long\def\mc#1{\mathscr{#1}}%
 
\global\long\def\ind{\mathbf{\mathbbm1}}%
\global\long\def\indep{\perp}%

\global\long\def\any{\forall}%
 
\global\long\def\ex{\exists}%
 
\global\long\def\p{\partial}%
 
\global\long\def\cd{\cdot}%
 
\global\long\def\Dif{\nabla}%
 
\global\long\def\imp{\Rightarrow}%
 
\global\long\def\iff{\Leftrightarrow}%

\global\long\def\up{\uparrow}%
 
\global\long\def\down{\downarrow}%
 
\global\long\def\arrow{\rightarrow}%
 
\global\long\def\rlarrow{\leftrightarrow}%
 
\global\long\def\lrarrow{\leftrightarrow}%
\global\long\def\gto{\rightarrow}%

\global\long\def\abs#1{\left|#1\right|}%
 
\global\long\def\norm#1{\left\Vert #1\right\Vert }%
 
\global\long\def\rest#1{\left.#1\right|}%

\global\long\def\bracket#1#2{\left\langle #1\middle\vert#2\right\rangle }%
 
\global\long\def\sandvich#1#2#3{\left\langle #1\middle\vert#2\middle\vert#3\right\rangle }%
 
\global\long\def\turd#1{\frac{#1}{3}}%
 
\global\long\def\ellipsis{\textellipsis}%
 
\global\long\def\sand#1{\left\lceil #1\right\vert }%
 
\global\long\def\wich#1{\left\vert #1\right\rfloor }%
 
\global\long\def\sandwich#1#2#3{\left\lceil #1\middle\vert#2\middle\vert#3\right\rfloor }%

\global\long\def\abs#1{\left|#1\right|}%
 
\global\long\def\norm#1{\left\Vert #1\right\Vert }%
 
\global\long\def\rest#1{\left.#1\right|}%
 
\global\long\def\inprod#1{\left\langle #1\right\rangle }%
 
\global\long\def\ol#1{\overline{#1}}%
 
\global\long\def\ul#1{\underline{#1}}%
 
\global\long\def\td#1{\tilde{#1}}%

\global\long\def\upto{\nearrow}%
 
\global\long\def\downto{\searrow}%
 
\global\long\def\pto{\overset{p}{\longrightarrow}}%
 
\global\long\def\dto{\overset{d}{\longrightarrow}}%
 
\global\long\def\asto{\overset{a.s.}{\longrightarrow}}%

\title{Logical Differencing in Dyadic Network Formation Models with Nontransferable
Utilities}
\author{Wayne Yuan Gao\thanks{Gao (corresponding author): Department of Economics, University of
Pennsylvania, 133 S 36th St., Philadelphia, PA 19104, USA, waynegao@upenn.edu.}, Ming Li\thanks{Li: Department of Economics, National University of Singapore, 1 Arts
Link AS2, Singapore 117570, mli@nus.edu.sg.}, and Sheng Xu\thanks{Xu: Department of Statistics and Data Science, Yale University, 24
Hillhouse Ave., New Haven, CT 06511, USA, sheng.xu@yale.edu.}}
\maketitle
\begin{abstract}
\noindent This paper considers a semiparametric model of dyadic network
formation under nontransferable utilities (NTU). Such dyadic links
arise frequently in real-world social interactions that require bilateral
consent but by their nature induce additive non-separability. In our
model we show how unobserved individual heterogeneity in the network
formation model can be canceled out without requiring additive separability.
The approach uses a new method we call \emph{logical differencing}.
The key idea is to construct an observable event involving the intersection
of two mutually exclusive restrictions on the fixed effects, while
these restrictions are as necessary conditions of weak multivariate
monotonicity. Based on this identification strategy we provide consistent
estimators of the network formation model under NTU. Finite-sample
performance of our method is analyzed in a simulation study, and an
empirical illustration using the risk-sharing network data from Nyakatoke
demonstrates that our proposed method is able to obtain economically
intuitive estimates.
\end{abstract}

\newpage{}

\section{\label{sec:C5_Intro}Introduction}

This paper considers a semiparametric model of dyadic network formation
under \emph{nontransferable utilities} (NTU), which arise naturally
in the modeling of real-world social interactions that require bilateral
consent. For instance, friendship is usually formed only when both
individuals in question are willing to accept each other as a friend,
or in other words, when both individuals derive sufficiently high
utilities from establishing the friendship. It is often plausible
that the two individuals may derive very different utilities from
the friendship for a variety of reasons: for example, one of them
may simply be more introvert than the other and derive lower utilities
from the friendship. In addition, there may not be a feasible way
to perfectly transfer utilities between the two individuals. Monetary
payments may not be customary in many social contexts, and even in
the presence of monetary or in-kind transfers, \emph{utilities} may
not be perfectly transferable through these feasible forms of transfers,
say, when individuals have different marginal utilities with respect
to these transfers.\footnote{See surveys by \citet{aumann1967survey}, \citet*{hart1985nontransferable}
and \citet*{mclean2002values} for discussions on the implications
of NTU on link (bilateral relationship) and group formation from a
micro-theoretical perspective.} Given the considerable academic and policy interest in understanding
the underlying drivers of network formation,\footnote{For example, the formation of friendship among U.S. high-school students
has been studied by a long line of literature, such as \citet*{moody2001race},
\citet*{currarini2009economic,currarini2010identifying}, \citet*{boucher2015structural},
\Citet{currarini2016simple}, \citet{xu2018diverse} among others. } it is not only theoretically interesting but also empirically relevant
to incorporate NTU in the modeling of network formation.

This paper contributes to the line of econometric literature on network
formation by introducing and incorporating \emph{nontransferable utilities}
into dyadic network formation models. Previous work in this line of
literature focuses primarily on case of \emph{transferable utilities},
as represented in \citet*{graham2017econometric}, which considers
a parametric model with homophily effects and individual unobserved
heterogeneity of the following form:

\begin{equation}
D_{ij}=\ind\left\{ w\left(X_{i},X_{j}\right)^{'}\b_{0}+A_{i}+A_{j}\geq\e_{ij}\right\} \label{eq:Graham}
\end{equation}
where $D_{ij}$ is an observable binary variable that denotes the
presence or absence of a link between individual $i$ and $j$, $w\left(X_{i},X_{j}\right)$
represents a (symmetric) vector of pairwise observable characteristics
specific to $ij$ generated by a known function $w$ of the individual
observable characteristics $X_{i}$ and $X_{j}$ of $i$ and $j$,
while $A_{i}$ and $A_{j}$ stand for unobserved individual-specific
degree heterogeneity and $\e_{ij}$ is some idiosyncratic utility
shock. Model \eqref{eq:Graham} essentially says that, if the (stochastic)\emph{
joint surplus} generated by a bilateral link $s_{ij}:=w\left(X_{i},X_{j}\right)^{'}\b_{0}+A_{i}+A_{j}-\e_{ij}$
exceeds the threshold zero, then the link between $i$ and $j$ is
formed. The model implicitly assumes that the link surplus can be
freely distributed among the two individuals $i$ and $j$, and that
bargaining efficiency is always achieved, so that the undirected link
is formed if and only if the link surplus is positive. Given this
specification, \citet*{graham2017econometric} provides consistent
and asymptotically normal maximum-likelihood estimates for the homophily
effect parameter $\b_{0}$, assuming that the exogenous idiosyncratic
pairwise shocks $\epsilon_{ij}$ are independently and identically
distributed with a logistic distribution. Recently, \citet*{Candelaria2016}
and \citet*{Toth2017} provide semiparametric generalizations of \citet*{graham2017econometric},
while \citet*{gao2018nonparametric} established nonparametric identification
of a class of index models that further generalize \eqref{eq:Graham}.

This paper, however, generalizes \citet*{graham2017econometric} along
a different direction, and seeks to incorporate the natural micro-theoretical
feature of NTU into this class of network formation models. To illustrate\footnote{Starting from Section \ref{sec:NetForm}, we consider a more general
specification than the illustrative model \eqref{eq:Graham} introduced
here.}, consider the following simple adaption of model \eqref{eq:Graham}
with two threshold-crossing conditions:

\begin{equation}
D_{ij}=\ind\left\{ w\left(X_{i},X_{j}\right)^{'}\b_{0}+A_{i}\geq\e_{ij}\right\} \cd\ind\left\{ w\left(X_{i},X_{j}\right)^{'}\b_{0}+A_{j}\geq\e_{ji}\right\} ,\label{eq:C4_ExNTU}
\end{equation}
where the unobserved individual heterogeneity $A_{i}$ and $A_{j}$
\emph{separately }enter into two different threshold-crossing conditions.
This formulation could be relevant to scenarios where $A_{i}$ represents
individual $i$'s own intrinsic valuation of a generic friend: for
a relatively shy or introvert person $i$, a lower $A_{i}$ implies
that $i$ is less willing to establish a friendship link, regardless
of how sociable the counterparty is. For simplicity, suppose for now
that $w\left(X_{i},X_{j}\right)\equiv{\bf 0}$, $\e_{ij}\sim_{iid}F$,
$\epsilon_{ji}\sim_{iid}F$ and $\epsilon_{ij}\perp\epsilon_{ji}$\footnote{For our general result, we do not require $\epsilon_{ij}\perp\epsilon_{ji}$,
nor the log-concavity of $F$. They are used here for illustration
purpose only. See a discussion after \eqref{eq:example ntu link}.}. Focusing completely on the effects of $A_{i}$ and $A_{j}$, it
is clear that the TU model \eqref{eq:Graham} implies that only the
sum of ``sociability'', $A_{i}+A_{j}$, matters: the linking probability
among pairs with $A_{i}=A_{j}=1$ (two moderately social persons)
should be exactly the same as the linking probability among pairs
with $A_{i}=2$ and $A_{j}=0$ (one very social person and one very
shy person), which might not be reasonable or realistic in social
scenarios. In comparison, the linking probability among pairs with
$A_{i}=2$ and $A_{j}=0$ is lower than the linking probability among
pairs with $A_{i}=A_{j}=1$ under the NTU model \eqref{eq:C4_ExNTU}
with i.i.d. $\epsilon_{ij}$ and $\epsilon_{ji}$ that follow any
log-concave distribution\footnote{A distribution is log-concave if $F\left(x\right)^{\lambda}F\left(y\right)^{1-\lambda}\leq F\left(\lambda x+\left(1-\lambda\right)y\right)$.
Many commonly used distributions, such as uniform, normal, exponential,
logistic, chi-squared distributions, are log-concave. See \citet{bagnoli2005log}
for more details on log-concave distributions from a microeconomic
theoretical perspective.}:
\begin{align*}
 & \E\left[\rest{D_{ij}}w\left(X_{i},X_{j}\right)\equiv{\bf 0},A_{i}=2,A_{j}=0\right]\\
=\  & F\left(0\right)F\left(2\right)\\
<\  & F\left(1\right)F\left(1\right)\\
=\  & \E\left[\rest{D_{ij}}w\left(X_{i},X_{j}\right)\equiv{\bf 0},A_{i}=A_{j}=1\right]
\end{align*}
This is intuitive given the observation that, under bilateral consent,
the party with relatively lower utility is the pivotal one in link
formation. Moreover, even though we maintain strict monotonicity in
the unobservable characteristics $A_{i}$ and $A_{j}$, the NTU setting
can still effectively incorporate homophily effects on unobserved
heterogeneity: given that $w\left(X_{i},X_{j}\right)\equiv{\bf 0}$
and $A_{i}+A_{j}=2$, the linking probability is effectively decreasing
in $\left|A_{i}-A_{j}\right|$ under log-concave $F$. Hence, by explicitly
modeling NTU in dyadic network formation, we can accommodate more
flexible or realistic patterns of conditional linking probabilities
and homophily effects that are not present under the TU setting.

However, the NTU setting immediately induces a key technical complication:
as can be seen explicitly in model \eqref{eq:C4_ExNTU}, the observable
indexes, $w\left(X_{i},X_{j}\right)^{'}\b_{0}$ and $w\left(X_{j},X_{i}\right)^{'}\b_{0}$,
and the unobserved heterogeneity terms ($A_{i}$ and $A_{j}$) are
no longer additively separable from each other. In particular, notice
that, even though the utility specification for each individual inside
each of the two threshold-crossing conditions in model \eqref{eq:C4_ExNTU}
remains completely linear and additive, the multiplication of the
two (nonlinear) indicator functions directly destroys both linearity
and additive separability, rendering inapplicable most previously
developed econometric techniques that arithmetically ``difference
out'' the ``two-way fixed effects'' $A_{i}$ and $A_{j}$ based
on additive separability.\footnote{Equivalently, one could write model \eqref{eq:C4_ExNTU} in an alternative
form as a ``single'' \emph{composite} threshold-crossing condition:
\[
D_{ij}=\ind\left\{ \min\left\{ w\left(X_{i},X_{j}\right)^{'}\b_{0}+A_{i}-\e_{ij},w\left(X_{j},X_{i}\right)^{'}\b_{0}+A_{j}-\e_{ji}\right\} \geq0\right\} ,
\]
where additive separability is again lost in this alternative formulation.}

Given this technical challenge, this paper proposes a new identification
strategy termed \emph{logical differencing}, which helps cancel out
the unobserved heterogeneity terms, $A_{i}$ and $A_{j}$, without
requiring additive separability but leveraging the logical implications
of \emph{multivariate monotonicity} in model \eqref{eq:C4_ExNTU}.
The key idea is to construct an observable event involving the intersection
of two mutually exclusive restrictions on the fixed effects $A_{i}$
and $A_{j}$, which logically imply an event that can be represented
without $A_{i}$ or $A_{j}$. Specifically, in the context of the
illustrative model \eqref{eq:C4_ExNTU} above, we start by considering
the event where a given individual $\ol i$ is \emph{more popular}
than another individual $\ol j$ among a group of individuals $k$
with observable characteristics $X_{k}=\ol x$ while $\ol i$ is simultaneously
\emph{less popular }than another individual $\ol j$ among a group
of individuals with a certain realization of observable characteristics
$\ul x$. This is the same as the conditioning event in \citet*{Toth2017}
and analogous to the tetrad comparisons made in \citet{Candelaria2016}.
However, instead of using arithmetic differencing to cancel out the
unobserved heterogeneity $A_{\ol i}$ and $A_{\ol j}$ as in \citet{Candelaria2016}
and \citet*{Toth2017}, we make the following logical deductions based
on the monotonicity of the conditional popularity of $\ol i$ in $w\left(X_{\ol i},\ol x\right)^{'}\b_{0}$
and $A_{\ol i}$. First, the event that $\ol i$ is \emph{more popular}
than another individual $\ol j$ among the group of individuals with
$X_{k}=\ol x$ implies that either $w\left(X_{\ol i},\ol x\right)^{'}\b_{0}>w\left(X_{\ol j},\ol x\right)^{'}\b_{0}$
or $A_{\ol i}>A_{\ol j}$, while the event that $\ol i$ is \emph{less
popular} than another individual $\ol j$ among a different group
of individuals with $X_{l}=\ul x$ implies that either $w\left(X_{\ol i},\ul x\right)^{'}\b_{0}<w\left(X_{\ol j},\ul x\right)^{'}\b_{0}$
or $A_{\ol i}<A_{\ol j}$. Second, when both events occur simultaneously,
we can logically deduce that either $w\left(X_{\ol i},\ol x\right)^{'}\b_{0}>w\left(X_{\ol j},\ol x\right)^{'}\b_{0}$
or $w\left(X_{\ol i},\ul x\right)^{'}\b_{0}<w\left(X_{\ol j},\ul x\right)^{'}\b_{0}$
must have occurred, because $A_{\ol i}>A_{\ol j}$ and $A_{\ol i}<A_{\ol j}$
cannot simultaneously occur. Intuitively, the ``switch'' in the
relative popularity of $\ol i$ and $\ol j$ among the two groups
of individuals with characteristics $\ol x$ and $\ul x$ cannot be
driven by individual unobserved heterogeneity $A_{\ol i}$ and $A_{\ol j}$,
and hence when we indeed observe such a ``switch'', we obtain a
restriction on the parametric indices $w\left(X_{\ol i},\ol x\right)^{'}\b_{0}$,
$w\left(X_{\ol i},\ol x\right)^{'}\b_{0}$, $w\left(X_{\ol i},\ul x\right)^{'}\b_{0}$,
and $w\left(X_{\ol j},\ul x\right)^{'}\b_{0}$, which helps identify
$\b_{0}$.

Based on this identification strategy we provide sufficient conditions
for point identification of the parameter $\b_{0}$ up to scale normalization
as well as a consistent estimator for $\b_{0}$. Our estimator has
a two-step structure, with the first step being a standard nonparametric
estimator of conditional linking probabilities, which we use to assert
the occurrence of the conditioning event, while in the second step
we use the identifying restriction on $\b_{0}$ when the conditioning
event occurs. The computation of the estimator essentially follows
the same method proposed in \citet*{gao2018robust}, with some adaptions
to the network data setting. We plot the identified sets under various
restrictions on the support of the observable characteristics $X_{i}$,
analyze the finite-sample performance in a simulation study, and present
an empirical illustration of our method using data from Nyakatoke
on risk-sharing network collected by Joachim De Weerdt.

~

This paper belongs to the line of literature that studies dyadic network
formation in a single large network setting, including \citet*{blitzstein2011},
\citet*{chatterjee2011}, \citet*{yan2013}, \citet*{yan2016}, \citet*{graham2017econometric},
\citet*{charbonneau2017}, \citet*{dzemski2017empirical}, \citet*{jochmans2017semiparametric},
\citet*{yan2018statistical}, \citet*{Candelaria2016}, \citet*{Toth2017}
and \citet*{gao2018nonparametric}. \citet*{shi2016structural} explicitly
incorporates NTU into dyadic network formation models, but \citet*{shi2016structural}
considers a fully parametric model and establishes the consistency
and asymptotic normality of the maximum likelihood estimators. See
also the recent surveys by \citet*{de2020econometric} and \citet*{graham2020dyadic}.

This paper is also related to a line of research that utilizes dyadic
link formation models in order to study structural social interaction
models: for instance, \citet{arduini2015parametric}, \citet*{auerbach2016identification},
\citet*{goldsmith2013social}, \citet*{hsieh2016social} and \citet*{johnsson2021estimation}.
In these papers, the social interaction models are the main focus
of identification and estimation, while the link formation models
are used mainly as a tool (a control function) to deal with network
endogeneity or unobserved heterogeneity problems in the social interaction
model. Even though some of the network formation models considered
in this line of literature is consistent with the NTU setting, this
line of literature is usually not primarily concerned with the full
identification and estimation of the network formation model itself.

It should be pointed out that in this paper we do not consider link
interdependence in network formation, which is studied by the line
of econometric literature on strategic network formation models. This
line of literature primarily uses pairwise stability \citep*{Jackson1996}
as the solution concept for network formation, and also often builds
NTU into the econometric specification. See, for example, \citet*{de2018identifying},
\citet*{graham2016homophily}, \citet*{leung2015random}, \citet*{menzel2015b},
\citet{boucher2017my}, \citet*{mele2017dense}, \citet*{mele2017structural}
and \citet*{ridder2017estimation}. However, this type of models usually
do not feature unobserved heterogeneity as in this paper. See, for
example, \citet*{de2020strategic} for a more detailed survey on this
line of literature.

This paper is also closely related to to \citet*{gao2018robust},
which similarly leverages multivariate monotonicity in a multi-index
structure under a panel multinomial choice setting. It should be pointed
out that, even though there is some structural similarity between
network data and panel data, there are no direct ways in the network
setting to make ``intertemporal comparison'' as in the panel setting,
which holds the fixed effects unchanged across two observable periods
of time. It is precisely this additional complication induced by the
network setting that requires the technique of logical differencing
proposed in this paper.

~

The rest of the paper is organized as follows. In Section \ref{sec:NetForm},
we describe the general specifications of our dyadic network formation
model. Section \ref{sec:ID_Est} establishes identification of the
parameter of interests in our model and provides a consistent tetrad
estimator. We plot the identified sets under various restrictions
on $X$ and report baseline simulation results in Section \ref{sec:C5_simu}.
We present an empirical illustration using the risk-sharing data of
Nyakatoke in Section \ref{sec:Emp}. Section \ref{sec:C5_conc} concludes.
Proofs and additional simulation results are available in the Appendix.

\section{\label{sec:NetForm}A Nonseparable Dyadic Network Formation Model}

We consider the following dyadic network formation model:
\begin{equation}
\E\left[\rest{D_{ij}}X_{i},X_{j},A_{i},A_{j}\right]=\phi\left(w\left(X_{i},X_{j}\right)^{'}\b_{0},A_{i},A_{j}\right)\label{eq:Model_NF}
\end{equation}
where:
\begin{itemize}
\item $i\in\left\{ 1,...,n\right\} $ denote a generic individual in a group
of $n$ individuals.
\item $X_{i}$ is a $\R^{d_{x}}$-valued vector of observable characteristics
for individual $i$. This could include, for example, wealth, age,
education and ethnicity of individual $i$.
\item $D_{ij}$ denotes a binary observable variable that indicates the
presence or absence of an undirected and unweighted link between two
distinct individuals $i$ and $j$: $D_{ij}=D_{ji}$ for all pairs
of individuals $ij$, with $D_{ij}=1$ indicating that $ij$ are linked
while $D_{ij}=0$ indicating that $ij$ are not linked.
\item $w:\R^{d_{x}}\times\R^{d_{x}}\to\R^{d}$ is a known function that
is \emph{symmetric}\footnote{Our method can also be adapted to the case with \emph{asymmetric}
$w$. See Remark \ref{rem:AymW}.} with respect to its two vector arguments. We will write $W_{ij}:=w\left(X_{i},X_{j}\right)$
for notational simplicity.
\item $\b_{0}\in\R^{d}$ is an unknown finite-dimensional parameter of interest.
Assume $\b_{0}\neq{\bf 0}$ so that we may normalize $\norm{\b_{0}}=1$,
i.e., $\b_{0}\in\S^{d-1}$.
\item $A_{i}$ is an unobserved scalar-valued variable that represents unobserved
individual heterogeneity.
\item $\phi:\R^{3}\to\R$ is an unknown measurable function that is symmetric
with respect to its second and third arguments.
\end{itemize}
In addition, we impose the following two assumptions:
\begin{assumption}[Monotonicity]
\label{assu:C5_mon} $\phi$ is weakly increasing in each of its
arguments.
\end{assumption}
Assumption \ref{assu:C5_mon} is the key assumption on which our identification
analysis is based. It requires that the conditional linking probability
between individuals with characteristics $\left(X_{i},A_{i}\right)$
and $\left(X_{j},A_{j}\right)$ be monotone in a parametric index
$\d_{ij}:=W_{ij}^{'}\b_{0}$ as well as the unobserved individual
heterogeneity terms $A_{i}$ and $A_{j}$. It should be noted that,
given monotonicity, increasingness is without loss of generality as
$\phi$, $\b_{0}$ and $A_{i},A_{j}$ are all unknown or unobservable.
In addition, Assumption \ref{assu:C5_mon} only requires that $\phi$
is monotonic in the index $W_{ij}^{'}\b_{0}$ as a whole, not individual
coordinates of $W_{ij}$. Therefore, we may include nonlinear or non-monotone
functions $w\left(\cd,\cd\right)$ on the observable characteristics
as long as Assumption \ref{assu:C5_mon} is maintained.

Next, we impose a standard random sampling assumption:
\begin{assumption}[Random Sampling]
\label{assu:DNF-rs} $\left(X_{i},A_{i}\right)$ is i.i.d. across
$i\in\left\{ 1,...,n\right\} $.
\end{assumption}
In particular, Assumption \ref{assu:DNF-rs} allows arbitrary dependence
structures between the observable characteristics $X_{i}$ and the
unobservable characteristic $A_{i}$.

~

Model \eqref{eq:Model_NF} along with the specifications and the two
assumptions introduced above encompass a large class of dyadic network
formation models in the literature. For example, the standard dyadic
network formation model \eqref{eq:Graham} studied by \citet*{graham2017econometric}
can be written as
\[
\E\left[\rest{D_{ij}}X_{i},X_{j},A_{i},A_{j}\right]=F\left(W_{ij}^{'}\b_{0}+A_{i}+A_{j}\right)
\]
where $F$ is the CDF of the standard logistic distribution. For the
semiparametric version considered by \citet{Candelaria2016}, \citet{Toth2017},
and \citet{gao2018nonparametric}, we can simply take $F$ to be some
unknown CDF. In either case, the monotonicity of the CDF $F$ and
the additive structure of $W_{ij}^{'}\b_{0}+A_{i}+A_{j}$ immediately
imply Assumption \ref{assu:C5_mon}.

However, our current model specification and assumptions further incorporate
a larger class of dyadic network formation models with potentially
nontransferable utilities. Specifically, consider the joint requirement
of two threshold-crossing conditions
\begin{align}
D_{ij} & =\ind\left\{ u\left(W_{ij}^{'}\b_{0},A_{i},A_{j},\e_{ij}\right)\geq0\right\} \cd\ind\left\{ u\left(W_{ji}^{'}\b_{0},A_{j},A_{i},\e_{ji}\right)\geq0\right\} ,\label{eq:Model_NFNTU}
\end{align}
where $u$ is an unknown function that is not necessarily symmetric
with respect to its second and third arguments $\left(A_{i},A_{j}\right)$,
and $\left(\e_{ij},\e_{ji}\right)$ are idiosyncratic pairwise shocks
that are i.i.d. across each unordered $ij$ pair with some unknown
distribution. In particular, notice that model \eqref{eq:C4_ExNTU}
is a special case of \eqref{eq:Model_NFNTU}. Suppose we further impose
the following two lower-level assumptions 1a and 1b:
\begin{assumption*}[\textbf{1a}]
 $\left(\e_{ij},\e_{ji}\right)$ are independent of $\left(X_{i},A_{i},X_{j},A_{j}\right)$.
\end{assumption*}
\begin{assumption*}[\textbf{1b}]
 $u$ is weakly increasing in its first three arguments.
\end{assumption*}
Then, the conditional linking probability
\begin{align}
 & \ \E\left[\rest{D_{ij}}X_{i},X_{j},A_{i},A_{j}\right]\nonumber \\
= & \ \int\ind\left\{ u\left(W_{ij}^{'}\b_{0},A_{i},A_{j},\e_{ij}\right)\geq0\right\} \cd\ind\left\{ u\left(W_{ji}^{'}\b_{0},A_{j},A_{i},\e_{ji}\right)\geq0\right\} d\P\left(\e_{ij},\e_{ji}\right)\nonumber \\
=: & \ \phi\left(W_{ij}^{'}\b_{0},A_{i},A_{j}\right)\label{eq:example ntu link}
\end{align}
can be represented by model \eqref{eq:Model_NF} with Assumption \ref{assu:C5_mon}
satisfied.

In particular, we do not require $\e_{ij}\indep\e_{ji}$. In fact,
$\e_{ij}\equiv\e_{ji}$ is readily incorporated in our model. Under
the maintained assumption that $w\left(X_{i},X_{j}\right)=w\left(X_{j},X_{i}\right)$,
if $\e_{ij}\equiv\e_{ji}$ and $u$ is furthermore assumed to be symmetric
with respect to its second and third arguments ($A_{i}$ and $A_{j}$),
then our model specializes to the case of transferable utilities,
\[
D_{ij}=\ind\left\{ u\left(W_{ij}^{'}\b_{0},A_{i},A_{j},\e_{ij}\right)\geq0\right\} ,
\]
where effectively only one threshold crossing condition determines
the establishment of a given network link. Therefore, our NTU model
\eqref{eq:Model_NF} includes the TU model as a special case.
\begin{rem}[Symmetry of $w$]
\label{rem:AymW} To explain the key idea of our identification strategy
in a notation-economical way, we will be focusing on the case of symmetric
$w$ in most of the following sections. However, it should be pointed
out that our method can also be applied to the case where $w$ is
allowed to be asymmetric in \eqref{eq:Model_NFNTU}, so that individual
utilities based on observable characteristics can also be made asymmetric
(nontransferable). In that case, model \eqref{eq:Model_NFNTU} needs
to be modified as
\begin{equation}
\E\left[\rest{D_{ij}}X_{i},X_{j},A_{i},A_{j}\right]=\phi\left(W_{ij}^{'}\b_{0},W_{ji}^{'}\b_{0},A_{i},A_{j}\right),\label{eq:Model_AsymW}
\end{equation}
where $W_{ij}=w\left(X_{i},X_{j}\right)$ may be different from $W_{ji}=w\left(X_{j},X_{i}\right)$,
but $\phi$ is symmetric with respect to its first two arguments $W_{ij}^{'}\b_{0},W_{ji}^{'}\b_{0}$\emph{
whenever} $A_{i}=A_{j}$ Moreover, Assumption \ref{assu:C5_mon} should
also be changed to be $\phi$ is monotone in all its \emph{four} arguments.
See Appendix \ref{subsec:ID_asym} for a more detailed discussion
on how our identification strategy can be adapted to accommodate asymmetric
$w$ under appropriate conditions.
\end{rem}

\section{\label{sec:ID_Est}Identification and Estimation}

\subsection{\label{subsec:DNF_ID}Identification via Logical Differencing}

In this section, we explain the key idea of our identification strategy.
We construct a mutually exclusive event to cancel out the unobservable
heterogeneity $A_{i}$ and $A_{j}$, which leads to an identifying
restriction on $\b_{0}$. We call this technique ``\textit{logical
differencing}''.

For each fixed individual $\ol i$, and each possible $\ol x\in\R^{d_{x}}$,
define
\begin{equation}
\rho_{\ol i}\left(\ol x\right):=\E\left[\rest{D_{\ol ik}}X_{k}=\ol x\right]\label{eq:define rho_i(x)}
\end{equation}
as the linking probability of this specific individual $\ol i$ with
a group of individuals, individually indexed by $k$, with the same
observable characteristics $X_{k}=\ol x$ (but potentially different
fixed effects $A_{k}$). Clearly, $\rho_{i}\left(\ol x\right)$ is
directly identified from data in a single large network.

Suppose that individual $\ol i$ has observed characteristics $X_{\ol i}=x_{\ol i}$
and unobserved characteristics $A_{\ol i}=a_{\ol i}$. Then, by model
\eqref{eq:Model_NF} we have
\begin{align}
\rho_{\ol i}\left(\ol x\right) & =\E\left[\rest{\E\left[\rest{D_{\ol ik}}X_{k}=\ol x,A_{k},X_{\ol i}=x_{\ol i},A_{\ol i}=a_{\ol i}\right]}X_{k}=\ol x\right]\nonumber \\
 & =\E\left[\rest{\phi\left(w\left(x_{\ol i},\ol x\right)^{'}\b_{0},a_{\ol i},A_{k}\right)}X_{k}=\ol x\right]\nonumber \\
 & =:\psi_{\ol x}\left(w\left(x_{\ol i},\ol x\right)^{'}\b_{0},a_{\ol i}\right),\label{eq:C5_rhorho}
\end{align}
where the expectation in the second to last line is taken over $A_{k}$
conditioning on $X_{k}=\ol x$. As we allow $A_{k}$ and $X_{k}$
to be arbitrarily correlated, the $\psi_{\ol x}$ function defined
in the last line of \eqref{eq:C5_rhorho} is dependent on $\ol x$.
In the same time, notice that $\psi_{\ol x}$ does not depend on the
identity of $\ol i$ beyond the values of $w\left(x_{\ol i},\ol x\right)^{'}\b_{0}$
and $a_{\ol i}$. By Assumption \ref{assu:C5_mon}, $\psi_{\ol x}\left(w\left(x_{\ol i},\ol x\right)^{'}\b_{0},a_{\ol i}\right)$
must be bivariate weakly increasing in the index $w\left(x_{\ol i},\ol x\right)^{'}\b_{0}$
and the unobserved heterogeneity scalar $a_{\ol i}$. We now show
how to use the bivariate monotonicity to obtain identifying restrictions
on $\b_{0}$.

~

Fixing two distinct individuals $\ol i$ and $\ol j$ in the population,
we first consider the event that \textit{individual} $\ol i$ is \emph{strictly
more popular than} \textit{individual} $\ol j$ among the \textit{group}
of individuals with observed characteristics $X_{k}=\ol x$:
\begin{equation}
\rho_{\ol i}\left(\ol x\right)>\rho_{\ol j}\left(\ol x\right),\label{eq:rho(ik) >jk}
\end{equation}
which is an event directly identifiable from observable data given
\eqref{eq:define rho_i(x)}. Even though event \eqref{eq:rho(ik) >jk}
is the same conditioning event as considered in \citet*{Toth2017}
and analogous to the tetrad comparisons made in \citet{Candelaria2016},
we now exploit the following logical deduction based on the bivariate
monotonicity of the conditional popularity of $\ol i$ in $w\left(X_{\ol i},\ol x\right)^{'}\b_{0}$
and $A_{\ol i}$ without the assumption of additivity between them.
Specifically, writing $\left(x_{\ol i},a_{\ol i}\right)$ and $\left(x_{\ol j},a_{\ol j}\right)$
as the observable and unobservable characteristics of individuals
$\ol i$ and $\ol j$, by \eqref{eq:C5_rhorho} we have
\begin{align}
 & \rho_{\ol i}\left(\ol x\right)>\rho_{\ol j}\left(\ol x\right).\nonumber \\
\iff\quad & \psi_{\ol x}\left(w\left(x_{\ol i},\ol x\right)^{'}\b_{0},a_{\ol i}\right)>\psi_{\ol x}\left(w\left(x_{\ol j},\ol x\right)^{'}\b_{0},a_{\ol j}\right)\nonumber \\
\imp\quad & \left\{ w\left(x_{\ol i},\ol x\right)^{'}\b_{0}>w\left(x_{\ol j},\ol x\right)^{'}\b_{0}\right\} \text{ OR }\left\{ a_{\ol i}>a_{\ol j}\right\} ,\label{eq:C5_contra}
\end{align}
Note that the last line of equation \eqref{eq:C5_contra} is a natural
necessary (but not sufficient) condition for $\rho_{\ol i}\left(\ol x\right)>\rho_{\ol j}\left(\ol x\right)$
under bivariate monotonicity.

~

Now, consider the event that \textit{individual} $\ol i$ is \emph{strictly
less popular than} \textit{individual} $\ol j$ among the \textit{group}
of individuals with observed characteristics $X_{k}=\underline{x}$,
i.e., 
\begin{equation}
\rho_{\ol i}\left(\ul x\right)<\rho_{\ol j}\left(\ul x\right).\label{eq:rho_i(l) < rho_j(l)}
\end{equation}
Then, by a similar argument to \eqref{eq:C5_contra}, we deduce
\begin{equation}
\rho_{\ol i}\left(\ul x\right)<\rho_{\ol j}\left(\ul x\right)\quad\imp\quad\left\{ w\left(x_{\ol i},\ul x\right)^{'}\b_{0}<w\left(x_{\ol j},\ul x\right)^{'}\b_{0}\right\} \text{ OR }\left\{ a_{\ol i}<a_{\ol j}\right\} .\label{eq:C5_contra2}
\end{equation}
Notice that the event $\left\{ a_{\ol i}<a_{\ol j}\right\} $ in \eqref{eq:C5_contra2}
is mutually exclusive with the event $\left\{ a_{\ol i}>a_{\ol j}\right\} $
that shows up in \eqref{eq:C5_contra}.

~

Next, consider the event that the two events \eqref{eq:rho(ik) >jk}
and \eqref{eq:rho_i(l) < rho_j(l)} described above\emph{ simultaneously
happen}. Then, by \eqref{eq:C5_contra}, \eqref{eq:C5_contra2} and
basic logical operations, we have
\begin{align}
 & \text{\ensuremath{\phantom{\text{ OR }}} }\left\{ \rho_{\ol i}\left(\ol x\right)>\rho_{\ol j}\left(\ol x\right)\right\} \text{ AND }\left\{ \rho_{\ol i}\left(\ul x\right)<\rho_{\ol j}\left(\ul x\right)\right\} \nonumber \\
\imp\quad & \text{\ensuremath{\phantom{\text{ OR }}} }\left(\left\{ w\left(x_{\ol i},\ol x\right)^{'}\b_{0}>w\left(x_{\ol j},\ol x\right)^{'}\b_{0}\right\} \text{ OR }\left\{ a_{\ol i}>a_{\ol j}\right\} \right)\nonumber \\
 & \ \text{AND\  \ensuremath{\left(\left\{  w\left(x_{\ol i},\ul x\right)^{'}\b_{0}<w\left(x_{\ol j},\ul x\right)^{'}\b_{0}\right\}  \text{ OR }\left\{  a_{\ol i}<a_{\ol j}\right\}  \right)}}\nonumber \\
\iff\quad & \phantom{\text{ OR }}\left(\left\{ w\left(x_{\ol i},\ol x\right)^{'}\b_{0}>w\left(x_{\ol j},\ol x\right)^{'}\b_{0}\right\} \ \text{AND}\ \left\{ w\left(x_{\ol i},\ul x\right)^{'}\b_{0}<w\left(x_{\ol j},\ul x\right)^{'}\b_{0}\right\} \right)\nonumber \\
 & \text{ OR }\left(\left\{ w\left(x_{\ol i},\ol x\right)^{'}\b_{0}>w\left(x_{\ol j},\ol x\right)^{'}\b_{0}\right\} \ \text{AND}\ \left\{ a_{\ol i}<a_{\ol j}\right\} \right)\nonumber \\
 & \text{ OR }\left(\left\{ a_{\ol i}>a_{\ol j}\right\} \ \text{AND}\ \left\{ w\left(x_{\ol i},\ul x\right)^{'}\b_{0}<w\left(x_{\ol j},\ul x\right)^{'}\b_{0}\right\} \right)\nonumber \\
 & \text{ OR }\left(\left\{ a_{\ol i}>a_{\ol j}\right\} \ \text{AND}\ \left\{ a_{\ol i}<a_{\ol j}\right\} \right)\nonumber \\
\imp\quad & \phantom{\text{ OR }}\left(\left\{ w\left(x_{\ol i},\ol x\right)^{'}\b_{0}>w\left(x_{\ol j},\ol x\right)^{'}\b_{0}\right\} \ \text{AND}\ \left\{ w\left(x_{\ol i},\ul x\right)^{'}\b_{0}<w\left(x_{\ol j},\ul x\right)^{'}\b_{0}\right\} \right)\nonumber \\
 & \text{ OR }\left\{ w\left(x_{\ol i},\ol x\right)^{'}\b_{0}>w\left(x_{\ol j},\ol x\right)^{'}\b_{0}\right\} \nonumber \\
 & \text{ OR }\left\{ w\left(x_{\ol i},\ul x\right)^{'}\b_{0}<w\left(x_{\ol j},\ul x\right)^{'}\b_{0}\right\} \nonumber \\
\iff\quad & \left\{ \left(w\left(x_{\ol i},\ol x\right)-w\left(x_{\ol j},\ol x\right)\right)^{'}\b_{0}>0\right\} \ \text{OR}\ \left\{ \left(w\left(x_{\ol i},\ul x\right)-w\left(x_{\ol j},\ul x\right)\right)^{'}\b_{0}<0\right\} ,\label{eq:C5_logicaldiff}
\end{align}

The derivations above exploit two simple logical properties: first,
\[
\left\{ a_{\ol i}>a_{\ol j}\right\} \ \text{AND}\ \left\{ a_{\ol i}<a_{\ol j}\right\} \quad=\quad\text{FALSE},
\]
and second,
\[
\left\{ w\left(x_{\ol i},\ol x\right)^{'}\b_{0}>w\left(x_{\ol j},\ol x\right)^{'}\b_{0}\right\} \ \text{AND}\ \left\{ a_{\ol i}<a_{\ol j}\right\} \quad\imp\quad\left\{ w\left(x_{\ol i},\ol x\right)^{'}\b_{0}>w\left(x_{\ol j},\ol x\right)^{'}\b_{0}\right\} ,
\]
which uses only necessary but not sufficient condition, so that we
can obtain an identifying restriction \eqref{eq:C5_logicaldiff} on
$\b_{0}$ that does not involve $a_{\ol i}$ nor $a_{\ol j}$. These
two forms of logical operations together enable us to ``difference
out'' (or ``cancel out'') the unobserved heterogeneity terms $a_{\ol i}$
and $a_{\ol j}$.

In contrast with various forms of ``\emph{arithmetic differencing}''
techniques proposed in the econometric literature (including \citealp{Candelaria2016}
and \citealp{Toth2017} specific to the dyadic network formation literature),
our proposed technique does \emph{not} rely on additive separability
between the parametric index $w\left(x_{\ol i},\ol x\right)^{'}\b_{0}$
and the unobserved heterogeneity term $a_{\ol i}$. Instead, our identification
strategy is based on multivariate monotonicity and utilizes logical
operations rather than standard arithmetic differencing to cancel
out the unobserved heterogeneity terms. Hence, we term our method
``\emph{logical differencing}''.

~

The identifying arguments above are derived for a fixed pair of individuals
$\ol i$ and $\ol j$, but clearly the arguments can be applied for
any pair of individuals $\left(i,j\right)$ with observable characteristics
$x_{i}$ and $x_{j}$. Writing
\begin{align*}
\tau_{ij}\left(\ol x,\ul x\right) & :=\ind\left\{ \rho_{i}\left(\ol x\right)>\rho_{j}\left(\ol x\right)\right\} \cd\ind\left\{ \rho_{i}\left(\ul x\right)<\rho_{j}\left(\ul x\right)\right\} \text{ }\text{and}\\
\l\left(\ol x,\ul x;x_{i},x_{j};\b\right) & :=\ind\left\{ \left(w\left(x_{i},\ol x\right)-w\left(x_{j},\ol x\right)\right)^{'}\b_{0}\leq0\right\} \cd\ind\left\{ \left(w\left(x_{i},\ul x\right)-w\left(x_{j},\ul x\right)\right)^{'}\b_{0}\geq0\right\} 
\end{align*}
for each $\b\in\S^{d-1}$, we summarize the identifying arguments
above by the following lemma.
\begin{lem}[Identifying Restriction]
\label{lem:DNF1}Under model \eqref{eq:Model_NF} and Assumptions
\ref{assu:C5_mon} and \ref{assu:DNF-rs}, we have

\begin{equation}
\tau_{ij}\left(\ol x,\ul x\right)=1\quad\imp\quad\l\left(\ol x,\ul x;x_{i},x_{j};\b_{0}\right)=0.\label{eq:ID_Rest}
\end{equation}
\end{lem}
A simple (but clearly not unique) way to build a criterion function
based on Lemma \ref{lem:DNF1} is to define

\begin{equation}
Q\left(\b\right):=\E_{ij,kl}\left[\tau_{ij}\left(X_{k},X_{l}\right)\l\left(X_{k},X_{l};X_{i},X_{j};\b\right)\right],\label{eq:pop_cri}
\end{equation}
where the expectation is $\E_{ij,kl}$ taken over random samples of
ordered tetrads $\left(i,j,k,l\right)$ from the population, and $\left(X_{i},X_{j},X_{k},X_{l}\right)$
denote the random variables corresponding to the observable characteristics
of $\left(i,j,k,l\right)$. According to Lemma \ref{lem:DNF1}, $Q\left(\b_{0}\right)=0$,
which is always smaller than or equal to $Q\left(\b\right)\geq0=Q\left(\b_{0}\right)$
for any $\b\neq\b_{0}$ because $\tau_{ij}\geq0$ and $\lambda\geq0$
by construction.

Observing that the scale of $\b_{0}$ is never identified, we write
\[
B_{0}:=\left\{ \b\in\S^{d-1}:Q\left(\b\right)=0\right\} 
\]
to represent the normalized ``identified set'' relative to the criterion
$Q$ defined in \eqref{eq:pop_cri}. Lemma \ref{lem:DNF1} implies
that $\b_{0}\in B_{0}$, but in general there is no guarantee that
$B_{0}$ is a singleton. The next subsection contains a set of sufficient
conditions that guarantees $B_{0}=\left\{ \b_{0}\right\} $.
\begin{rem}
We should point out that the identified set $B_{0}$ defined above
based on logical differencing is not \emph{sharp} in general, since
the individual unobserved heterogeneity term $A_{i}$ is \emph{canceled
out by logical differencing. }In fact, one can show $A_{i}$ is also
identified (up to proper normalization) under certain conditions\footnote{The identification of $A_{i}$ is conceptually analogous to the identification
of individual fixed effects in a long panel setting.}, and knowledge about $A_{i}$ can help with the identification of
$\b_{0}$. In fact, when all realizations of $A_{i}$ are point identified,
the point identification of $\b_{0}$ can be established under much
weaker conditions than those to be presented in Section \ref{subsec:PointID}.\footnote{See \citet*{gao2018nonparametric} for a related discussion. In fact,
the identification strategy in \citet*{gao2018nonparametric} can
be adapted to establish identification of $A_{i}$ under the NTU setting.
However, a rigorous presentation of such identification results is
beyond the scope of this paper.} However, we trade sharpness for simplicity: this paper provides a
method of identification (and estimation) without the need to deal
with the incidental parameters $A_{i}$.
\end{rem}
~

It is worth mentioning that for any one-sided sign preserving function
$\g$ such that 
\begin{equation}
\g\left(t\right)\ \begin{cases}
\geq0, & \text{for }t>0,\\
=0, & \text{for }t\leq0.
\end{cases}\label{eq:gamma_sign}
\end{equation}
we may define
\begin{align}
\tau_{ij}^{\g}\left(\ol x,\ul x\right) & :=\g\left(\rho_{i}\left(\ol x\right)-\rho_{j}\left(\ol x\right)\right)\cd\g\left(\rho_{j}\left(\ul x\right)-\rho_{i}\left(\ul x\right)\right),\label{eq:smooth tau}\\
Q^{\g}\left(\b\right) & :=\E_{ij,kl}\left[\tau_{ij}^{\g}\left(X_{k},X_{l}\right)\l\left(X_{k},X_{l};X_{i},X_{j};\b\right)\right],\label{eq:smooth_Q}
\end{align}
without changing the identification set at all, since 
\[
\tau_{ij}\left(\ol x,\ul x\right)>0\quad\text{if and only if}\quad\tau_{ij}^{\g}\left(\ol x,\ul x\right)>0.
\]
In fact, $\tau^{\g}$ specializes to $\tau$ when we set $\g\left(t\right):=\ind\left\{ t>0\right\} $.
Alternatively, we may set $\g$ to be ``smoother'', say, $\g\left(t\right):=\left[t\right]_{+}$
the positive part function.

Such forms of ``smoothing'' in the population criterion will be
irrelevant to all the identification results and its proofs in this
paper, so for notational simplicity, we will suppress $\g$ and focus
on the representative $\tau_{ij}$ and $Q$ in the next subsection
about point identification. However, a smooth $\g$ will play a role
when it comes to estimation and computation, and we will revisit $\g$
in Section \ref{subsec:DNF_est}.

\subsection{\label{subsec:PointID}Sufficient Conditions for Point Identification}

We now present a set of sufficient conditions that guarantee point
identification of $\b_{0}$ on $\S^{d-1}$.

\medskip{}

\noindent \textbf{Assumption 1$'$} \customlabel{as1prime}{$1'$}
(Strict Monotonicity of $\phi$)\textbf{.} $\phi$ defined in model
\eqref{eq:Model_NF} is weakly increasing in $A_{i}$ and $A_{j}$
while strictly increasing in the index $w\left(X_{i},X_{j}\right)^{'}\b_{0}$.

\medskip{}
Assumption \ref{as1prime} strengthens Assumption \ref{assu:C5_mon}
by requiring that $\phi$ be strictly increasing in the parametric
index $w\left(X_{i},X_{j}\right)^{'}\b_{0}$. This is used to guarantee
that differences in the parametric index can indeed lead to changes
in conditional linking probabilities, so that the conditional event
$\left\{ \tau_{ij}\left(\ol x,\ul x\right)=1\right\} $ in Lemma \ref{lem:DNF1}
may occur with strictly positive probability.
\begin{assumption}[Continuity of $\phi$ and $w$]
\label{assu:DNF-con} $\phi$ and $w$ are continuous functions on
their domains.
\end{assumption}
\begin{assumption}[Sufficient Directional Variations]
\label{assu:DNF-supportX} There exist distinct points $\ol x$ and
$\ul x$ in $Supp\left(X_{i}\right)$, such that the vector ${\bf 0}$
lies in the interior of $Supp\left(w\left(\ol x,X_{i}\right)-w\left(\ul x,X_{i}\right)\right)$.
\end{assumption}
We also provide a lower-level condition for Assumption \ref{assu:DNF-supportX}
when $w$ is the coordinate-wise Euclidean distance function.\medskip{}

\noindent \textbf{Assumption 4$'$}\customlabel{as3prime}{$4'$}\emph{Suppose
that (i) $w_{h}\left(\ol x,\ul x\right):=\left|\ol x_{h}-\ul x_{h}\right|$
for every coordinate $h$, and (ii) $Supp\left(X_{i}\right)$ has
nonempty interior.}\medskip{}

Essentially, since our criterion function is based on indicator functions
of halfspaces in the form of $\left(w\left(\ol x,\tilde{x}\right)-w\left(\ul x,\tilde{x}\right)\right)^{'}\b\lesseqqgtr0$,
we will need the distribution of these indicators to take both values
$0$ and $1$ with strictly positive probabilities under any $\b$,
so that every possible $\b$ different from $\b_{0}$ can be differentiated
from $\b_{0}$ by the criterion function (up to scale normalization).
Hence, we need sufficient variations in the observable covariates.

Assumption \ref{assu:DNF-supportX}, though apparently not very transparent
on its own, is actually implied by Assumption \ref{as3prime}.\footnote{\emph{Proof:} Suppose that $Supp\left(X_{i}\right)$ has nonempty
interior. Then there exist two distinct points $\ol x$ and $\ul x$
in the interior of $Supp\left(X_{i}\right)$ such that $\tilde{x}:=\frac{1}{2}\left(\ol x+\ul x\right)$
is also in the interior of $Supp\left(X_{i}\right)$. Clearly, $w\left(\ol x,\tilde{x}\right)=w\left(\ul x,\tilde{x}\right)=\left(\frac{1}{2}\left|\ol x_{1}-\ul x_{1}\right|,...,\frac{1}{2}\left|\ol x_{k}-\ul x_{k}\right|\right)^{'}$.
Since $\ol x,\ul x$ are all interior points of $Supp\left(X_{i}\right)$
and $w$ is continuous, the vector ${\bf 0}$ must be an interior
point of $Supp\left(w\left(\ol x,X_{i}\right)-w\left(\ul x,X_{i}\right)\right)$.
\qedsymbol} We note that the assumption of nonempty interior is a familiar one,
which is often imposed for point identification in the literature,
say, on maximum score estimation. Assumption \ref{as3prime} allows
the support of all observable covariates to be bounded, but on the
other hand require all covariates to be continuously distributed.

In Appendix \ref{subsec:PID_Discrete}, we present an alternative
set of assumptions that allow for the presence of discrete covariates,
but require the existence of a ``special covariate'' with large
(conditional) continuous support and nonzero coefficient, a la \citet*{horowitz1992smoothed}.
Since the identification result with special covariate needs to be
presented under a different scale normalization, we defer the results
to the appendix.
\begin{assumption}[Conditional Support of $A_{i}$]
\label{assu:DNF-supportA}$A_{i}$ is conditionally distributed on
the same support given $X_{i}=x$ for any realization $x\in Supp\left(X_{i}\right)$.
\end{assumption}
Assumption \ref{assu:DNF-supportA} together with Assumption \ref{assu:DNF-rs}
implies that for two randomly sampled individuals $i,j$, there is
a strictly positive probability of $A_{i}$ and $A_{j}$ being sufficiently
close to each other, conditional on any realizations of $X_{i}$ and
$X_{j}$. Together with the continuity condition in Assumption \ref{assu:DNF-con},
we can ensure that the parametric index based on observable covariates
alone can determine whether $i$ or $j$ is relatively more popular
among a certain group of individuals.

$\ $

Next, we lay out the lemma that will be used in the proof of point
identification of $\b_{0}$.
\begin{lem}[Differentiating $\b$ from $\b_{0}$]
\label{lem:DNF2} Under model \eqref{eq:Model_NF}, Assumptions \ref{as1prime},
\ref{assu:DNF-rs}--\ref{assu:DNF-supportA}, for each $\b\in\text{\ensuremath{\mathbb{S}}}^{d-1}\backslash\left\{ \b_{0}\right\} $,
\begin{align}
\tau_{ij}\left(X_{k},X_{l}\right) & =1,\label{eq:DNF_ID1}\\
\l\left(X_{k},X_{l};X_{i},X_{j};\b_{0}\right) & =0,\label{eq:DNF_ID2}\\
\l\left(X_{k},X_{l};X_{i},X_{j};\b\right) & =1,\label{eq:DNF_ID3}
\end{align}
occur simultaneously with strictly positive probability, as we randomly
sample individuals $i,j,k,l$.
\end{lem}
For point identification of $\b_{0}$, we need the population criterion
\eqref{eq:pop_cri} to differentiate each $\b\in\text{\ensuremath{\mathbb{S}}}^{d-1}\backslash\left\{ \b_{0}\right\} $
from $\b_{0}$. Lemma \ref{lem:DNF2} ensures this by establishing
that, the conditioning event \eqref{eq:DNF_ID1} in Lemma \ref{lem:DNF1}
occurs with positive probability, and, when it occurs, we can obtain
different values of $\l$ at $\b$ from that at $\b$ with positive
probabilities. Compared with the set identification result discussed
in Section \ref{subsec:DNF_ID}, we need to strengthen Assumption
\ref{assu:C5_mon} to \emph{strict} monotonicity in the index $w\left(X_{i},X_{j}\right)^{'}\b_{0}$.
Then, Assumptions \ref{assu:DNF-con} and \ref{assu:DNF-supportA}
guarantee that, when the absolute difference between $A_{i}$ and
$A_{j}$ get sufficiently small, differences in the parametric indexes
can lead to differences in conditional linking probabilities, so that
\eqref{eq:DNF_ID1} will occur. Since $\b_{0}$ and $\b$ define different
intersections of halfspaces for the vector $w\left(\ol x,\tilde{x}\right)-w\left(\ul x,\tilde{x}\right)$
through $\l$, Assumption \ref{assu:DNF-supportX} then guarantees
that there will be on-support realizations of the observable covariates
that help ``detect'' such differences.

We are now ready to present the point identification result.
\begin{thm}[Point Identification of $\b_{0}$]
\label{thm:DNF-ID}Under model \eqref{eq:Model_NF} and Assumptions
\ref{as1prime}, \ref{assu:DNF-rs}--\ref{assu:DNF-supportA}, $\b_{0}$
is the unique minimizer of $Q\left(\b\right)$ defined in \eqref{eq:pop_cri}
over the unit sphere $\text{\ensuremath{\mathbb{S}}}^{d-1}$. Furthermore,
for any $\epsilon>0$, there exists $\d>0$ such that 
\[
\inf_{\b\in\text{\ensuremath{\mathbb{S}}}^{d-1}\backslash B_{\e}\left(\b_{0}\right)}Q\left(\b\right)\geq Q\left(\b_{0}\right)+\d,
\]
where $B_{\e}\left(\b_{0}\right):=\left\{ \b\in\S^{d-1}:\|\b-\b_{0}\|\leq\e\right\} $.
\end{thm}
Theorem \ref{thm:DNF-ID} follows from Lemma \ref{lem:DNF2} based
on the standard arguments in \citet*{newey1994asymp}.
\begin{rem*}[Asymmetry of $w$, Continued]
In Appendix \ref{subsec:ID_asym}, we show how the identification
arguments and assumptions above can be adapted to accommodate asymmetry
of $w$. In short, the technique of logical differencing applies without
changes, but the identifying restriction we obtained becomes weaker.
In particular, when $w$ is \emph{antisymmetric} in the sense that
$w\left(\ol x,\ul x\right)+w\left(\ul x,\ol x\right)\equiv0$, the
identifying restriction we obtained through logical differencing becomes
trivial, and $B_{0}=\S^{d-1}$. However, with asymmetric but not antisymmetric
$w$, it is still feasible to strengthen Assumption 3 so as to obtain
point identification. See more discussions in Appendix \ref{subsec:ID_asym}.
\end{rem*}

\subsection{\label{subsec:DNF_est}Tetrad Estimator and Consistency}

We now proceed to present a consistent estimator of $\b_{0}$ in the
framework of extremum estimation, which we construct using a two-step
semiparametric estimation procedure. We clarify that we are considering
the asymptotics under ``a single large network'' with the number
of individuals $N\to\infty$. Moreover, we focus on the ``dense network''
asymptotics where the conditional linking probabilities are nondegenerate
in the limit and can be consistently estimated.

The first step is the nonparametric estimation of 
\[
\rho_{i}\left(x\right):=\E\left[D_{ik}\big|i,X_{k}=x\right].
\]
To implement this, we fix an individual $i$ in the sample, and regress
$D_{ik}$, the indicator function for the link between $i$ and $k$,
on the basis functions chosen by the researcher evaluated at observable
characteristics $X_{k}$ for all $k\neq i$. To guarantee that a consistent
nonparametric estimator of $\rho_{i}$ exists, we need to impose some
regularity conditions on $\rho_{i}$. We state the following assumption
as an illustrative set of such conditions, acknowledging that there
may be many different versions that also work.
\begin{assumption}[Regularity Conditions for $\rho_{i}$]
\label{assu:uniform-reg} (i) $Supp\left(X_{i}\right)$ is bounded
and convex with nonempty interior; (ii) for each fixed $i$, $\rho_{i}\in{\cal C}_{M}^{d_{x}+1}\left(Supp\left(X_{i}\right)\right)$,
where ${\cal C}_{M}^{d_{x}+1}\left(Supp\left(X_{i}\right)\right)$
denotes the class of functions on $Supp\left(X_{i}\right)$ whose
derivatives are uniformly bounded by $M$ up to order $d_{x}+1$.
\end{assumption}
Assumption \ref{assu:uniform-reg} essentially requires that $\rho_{i}\left(x\right)$
is smooth enough in $x$. Given that
\[
\rho_{i}\left(x\right)=\int\phi\left(w\left(\ol x,x\right)^{'}\b_{0},a_{i},A_{k}\right)d\P\left(\rest{A_{k}}X_{k}=x\right)
\]
is an integral of $\phi$ (a strictly increasing function bounded
between 0 and $1$) over the conditional distribution of $A_{k}$,
Assumption \ref{assu:uniform-reg} is easily satisfied, say, if both
$\phi$ and the conditional density of $A_{k}$ given $X_{k}=x$ have
uniformly bounded derivatives up to order $d+1$, when $w$ is taken
to be the coordinate-wise Euclidean distance function.
\begin{lem}
\label{lem:est_1step}Given Assumption \ref{assu:uniform-reg}, for
each $i$, there exists an estimator $\hat{\rho}_{i}\in{\cal C}_{M}^{d_{x}+1}\left({\cal X}\right)$
that is $L_{2}\left(\P_{X}\right)$ consistent, i.e., 
\[
\norm{\hat{\rho}_{i}-\rho_{i}}_{L_{2}\left(\P_{X}\right)}:=\sqrt{\int\left(\hat{\rho}_{i}\left(x\right)-\rho_{i}\left(x\right)\right)^{2}d\P_{X_{k}}\left(x\right)}=o_{p}\left(1\right).
\]
\end{lem}
Lemma \ref{lem:est_1step} follows from the large literature on many
different types of consistent nonparametric estimators. See \citet{Bierens1983kernel}
for results on kernel estimators and \citet{chen2007sieve} on sieve
estimators. In our simulation and application, we use a spline-based
sieve estimator.\medskip{}

In the second step, we use $\hat{\rho}_{i}$ to build the following
sample analog of the population criterion $Q^{\g}\left(\b\right)$
in \eqref{eq:pop_cri}:

\begin{equation}
\begin{aligned}\widehat{Q}_{n}^{\g}\left(\b\right):=\frac{\left(n-4\right)!}{n!}\sum_{i,j,k,l}\g\left(\hat{\rho}_{i}\left(X_{k}\right)-\hat{\rho}_{j}\left(X_{k}\right)\right)\g\left(\hat{\rho}_{j}\left(X_{l}\right)-\hat{\rho}_{i}\left(X_{l}\right)\right)\l\left(X_{k},X_{l};X_{i},X_{j};\b\right).\end{aligned}
\label{eq:sampeQ}
\end{equation}
The two-step tetrad estimator for $\b_{0}$ is then defined as

\begin{equation}
\widehat{\b}_{n}:=\arg\min_{\b\in\text{\ensuremath{\mathbb{S}}}^{d-1}}\widehat{Q}_{n}^{\g}\left(\b\right).\label{eq:beta^hat}
\end{equation}

Since $\g$ can be arbitrarily chosen as long as it preserves strict
positiveness as in Section \ref{subsec:DNF_est}, we now impose the
following continuity assumption for $\g$.
\begin{assumption}[Continuity of $\g$]
\label{assu:as_gamma} The one-sided sign-preserving function $\g$
is Lipschitz-continuous.
\end{assumption}
Assumption \ref{assu:as_gamma} can be achieved by setting, say, $\g\left(t\right):=\left[t\right]_{+}$
or $\g\left(t\right):=2\times\Phi\left(\left[t\right]_{+}\right)-1$,
where $\left[t\right]_{+}$ is the positive part of $t$, and $\Phi$
is the CDF of the standard normal distribution. See Section \ref{sec:C5_simu}
for details. The idea is that, when the difference $\rho_{i}\left(x\right)$
and $\rho_{j}\left(x\right)$ is small, the estimation of whether
$\rho_{i}\left(x\right)>\rho_{j}\left(x\right)$ may be relatively
imprecise, and therefore, we may wish to downweight such terms in
the criterion function.

Computationally, to exploit the topological characteristics of the
parameter space $\S^{d-1}$, i.e. compactness and convexity, we develop
a new bisection-style nested rectangle algorithm that recursively
shrinks and refines an adaptive grid on the angle space. The key novelty
of the algorithm is that instead of working with the edges of the
Euclidean parameter space $\R^{d}$, we deterministically ``cut''
the angle space in each dimension of $\S^{d-1}$ to search for the
area that minimizes $\widehat{Q}_{n}^{\g}(\b)$. Additional measures
are taken to ensure the search algorithm is conservative. Simulation
and empirical results show that our algorithm performs reasonably
well with a relatively small sample size. \citet{gao2018robust} provides
more details regarding the implementation in a panel multinomial choice
setting.

We now state the consistency result of the tetrad estimator $\widehat{\b}_{n}$
under point identification.
\begin{thm}
\label{thm:DNF-consistent}Under model \eqref{eq:Model_NF} and Assumptions
\ref{as1prime}, \ref{assu:DNF-rs}--\ref{assu:as_gamma}, $\hat{\b}_{n}$
is consistent for $\b_{0}$, i.e.,
\[
\widehat{\beta}_{n}\pto\b_{0}.
\]
\end{thm}
\begin{rem}
Our estimator shares some similarity with the maximum score estimator:
the parameter $\b$ enters into the sample criterion through indicators
of halfspaces about $\b$, which creates discreteness in the sample
criterion. However, our estimator features an additional complication
not found in maximum score estimation: we require the first-stage
nonparametric estimators $\hat{\rho}_{i}$, and we need to plug the
first-stage estimators into nonlinear functions that isolate the ``positive
side'' only: $\g\left(\hat{\rho}_{i}\left(X_{k}\right)-\hat{\rho}_{j}\left(X_{k}\right)\right)\g\left(\hat{\rho}_{j}\left(X_{l}\right)-\hat{\rho}_{i}\left(X_{l}\right)\right)$
will be nonzero only when $\hat{\rho}_{i}\left(X_{k}\right)>\hat{\rho}_{j}\left(X_{k}\right)$
and $\hat{\rho}_{j}\left(X_{l}\right)>\hat{\rho}_{i}\left(X_{l}\right)$.
In fact, it is precisely the ``double-threshold'' feature of our
model that leads to the necessity of a semiparametric two-step estimation
procedure, relative to the one-step procedure of the maximum score
estimator under a ``single-threshold'' setting. Given that the asymptotic
theory for the maximum score estimator is already nonstandard (with
cubic-root rate of convergence and non-normal Chernoff-type asymptotic
distribution as in \citealp{kim1990cube}), the addition of the first-stage
nonparametric estimation further complicates the asymptotic theory
to a highly nontrivial extent. The first-stage nonparametric estimation
of $\rho_{i}$ may further slow the rate of convergence below $n^{-1/3}$;
in the meanwhile, if we take $\g$ to be smooth (but necessarily nonlinear),
the term $\g\left(\hat{\rho}_{i}\left(X_{k}\right)-\hat{\rho}_{j}\left(X_{k}\right)\right)\g\left(\hat{\rho}_{j}\left(X_{l}\right)-\hat{\rho}_{i}\left(X_{l}\right)\right)$
may also provide some effective smoothing on the discrete $\l_{ij}$
term, when $\rho_{i}\left(X_{k}\right)-\rho_{j}\left(X_{k}\right)$
or $\rho_{j}\left(X_{l}\right)-\rho_{i}\left(X_{l}\right)$ is closer
to zero. It is not exactly clear which effect dominates, or whether
the two effects can be balanced, in our current setting. Due to such
technical difficulties, we defer the investigation of such types of
``two-stage maximum score estimators'' in a separate paper by \citet*{gao2020two},
albeit in a simpler setting.
\end{rem}

\section{\label{sec:C5_simu}Simulation}

In this section, we conduct a simulation study to analyze the finite-sample
performance of our two-step tetrad estimator. To begin with, we calculate
and plot the identified set $B_{0}$ via \eqref{eq:pop_cri} for various
support restrictions on $X$\footnote{We thank the Editor for this suggestion.}.
The graph illustrates how the size of the identified set $B_{0}$
based on our population criterion \eqref{eq:pop_cri} changes when
the support of $X$ contains more and more discrete variables. Then,
we specify the data generating process (DGP) of the Monte Carlo simulations.
We show and discuss the performance of our two-step estimation method
under the baseline setup with symmetric $w\left(\cd,\cd\right)$ function.
In Appendix \ref{subsec:Sim_Robustness}, we vary the number of individuals
$N$, the dimension of the pairwise observable characteristics $d$,
and the degree of correlation between $X$ and $A$, as well as allow
for asymmetric $w$ function to further examine the robustness of
our method.

\subsection{Identified Set $B_{0}$\label{subsec:Identified-Set}}

In this section, we provide a graphical illustration of the identified
set $B_{0}$ for various support restrictions on $X$. The analysis
is based on the following network formation model:
\begin{equation}
D_{ij}=\ind\left\{ w\left(X_{i},X_{j}\right)^{'}\b_{0}+A_{i}>\epsilon_{ij}\right\} \cd\ind\left\{ w\left(X_{j},X_{i}\right)^{'}\b_{0}+A_{j}>\epsilon_{ji}\right\} ,\label{eq:simrule}
\end{equation}
where $w$ is taken to be the coordinate-wise absolute difference,
i.e., $w_{h}\left(\ol x,\ul x\right):=\left|\ol x_{h}-\ul x_{h}\right|$
for $h=1,...,d$. We set $d_{x}=d=3$ and $\beta_{0}=\left(1,1,1\right)/6^{'}$\footnote{Recall that only the direction of $\b_{0}$ is identified. Here we
divide the vector $\left(1,1,1\right)^{'}$ by 6 to ensure the network
is non-degenerate numerically, i.e., not all agents are connected,
nor are they all disconnected.}. We incorporate correlation between $X$ and $A$ by drawing $A_{i}=\left(X_{i,1}+\xi_{i}\right)/4$,
where $\xi_{i}$ is independently and uniformly distributed on $\left[-0.5,0.5\right]$
and independent of all other variables. $\epsilon_{ij}$ is the exogenous
random shock independent of all other variables and uniformly distributed
on $\left[0,1\right]$.

The purpose of the exercise is to show how the support restrictions
on $X$ affect the size of the identified set $B_{0}$. To this end,
we consider the following five support conditions on $X$: (1) all
coordinates of $X$ are uniformly distributed on $[-0.5,0.5]$; (2)
$X_{i,1}$ is binary with equal probability on $\left\{ 0,1\right\} $,
while $X_{i,2}$ and $X_{i,3}$ are both uniformly distributed on
$[-0.5,0.5]$; (3) $X_{i,1}$ is binary with equal probability on
$\left\{ 0,1\right\} $, $X_{i,2}$ is discrete with equal probability
on 11 points of $\left\{ -0.5,-0.4,..,0,..,0.4,0.5\right\} $, and
$X_{i,3}$ is uniformly distributed on $[-0.5,0.5]$; (4) all coordinates
of $X$ are discrete with equal probability on 101 points of $\left\{ -0.5,-0.49,-0.48..,0,..,0.49,0.5\right\} $;
(5) all coordinates of $X$ are discrete with equal probability on
11 points of $\left\{ -0.5,-0.4,..,0,..,0.4,0.5\right\} $.

$\b_{0}$ is point identified in (1) since the support of $X$ has
a nonempty interior, thus satisfying Assumption \ref{as3prime}. In
(2)--(5), $\b_{0}$ is not point identified due to discreteness and
boundedness of the support of $X$. Below we compare the sizes of
the identified sets when the $X$ vector contains none, one, two and
all discrete variables corresponding to DGP (1)--(5).

To calculate the ID set, we utilize the analytical formula of $\rho_{i}$
so that the true $\rho_{i}$ is calculated without error \footnote{We use the distribution of $\left(A,\epsilon,\xi\right)$ to obtain
the analytical formula for $\rho_{i}\left(x\right)$.}. Then, we numerically approximate the population criterion by setting
a large $N=1,000$ and $M=10,000$. We are able to work with a larger
$N$ and $M$ than in the simulations in the next subsection because
to get the identified set, we only calculate the maximizer of the
population criterion once without the need to estimate $\rho_{i}\left(x\right)$.
Still, one can improve the numerical results by increasing $N,M\gto\infty$
and extracting more info from the data when additional computational
power is available. Hence, all our results below are conservative:
the true ID set must be a subset of the result shown in Figure \ref{fig:id_set_plot}.

For clarity of illustration, we plot the results in the angle space\footnote{One can transform any $\b=\left(\b_{1},\b_{2},\b_{3}\right)^{'}$
on the unit sphere $\S^{2}$ into the angle space by letting $\b_{1}=\cos\t_{1}\cos\t_{2},\b_{2}=\cos\t_{1}\sin\t_{2},\b_{3}=\sin\t_{1}$
for $\left(\t_{1},\t_{2}\right)\in\left[-\pi/2,\pi/2\right]\times\left[-\pi,\pi\right]$.}. In all cases, we maintain the correlation between $X$ and $A$.
The results are summarized in Figure \ref{fig:id_set_plot}, where
we plot the identified sets $B_{0}$ on the full parameter space of
$\left[-\pi/2,\pi/2\right]\times\left[-\pi,\pi\right]$ and zoom it
in for more details.
\noindent \begin{center}
\begin{figure}
\noindent \begin{centering}
\includegraphics[scale=0.75]{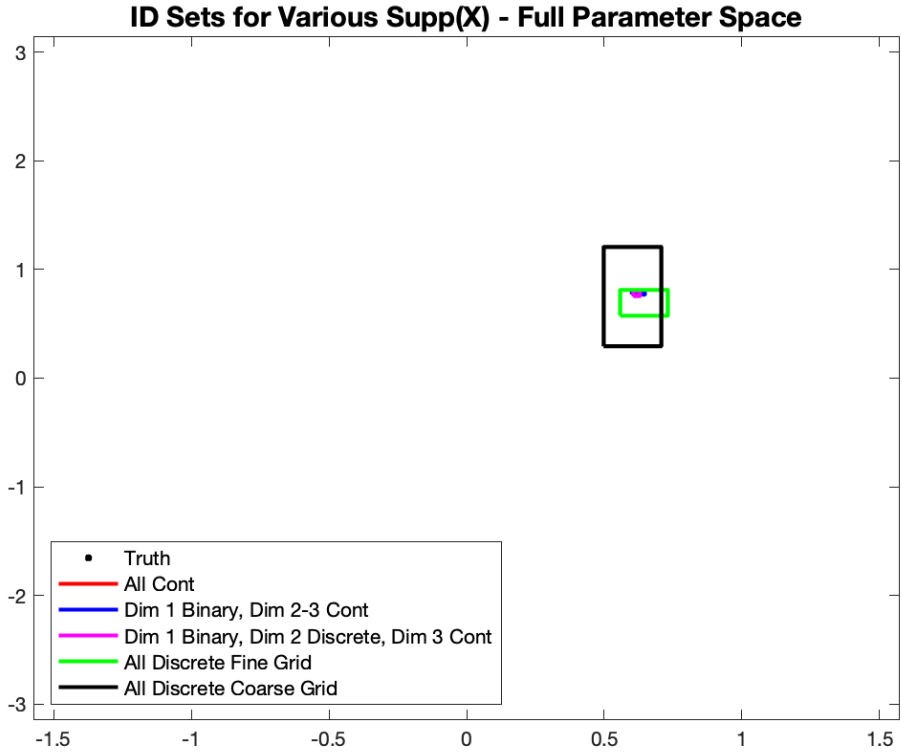}
\par\end{centering}
\noindent \centering{}\includegraphics[scale=0.75]{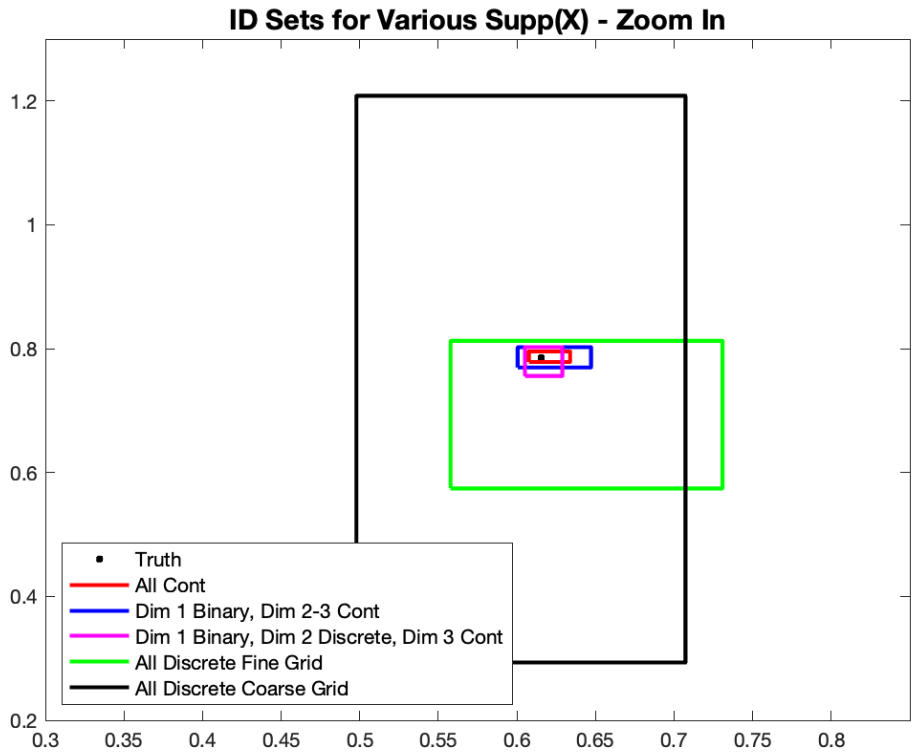}\caption{Identified Sets under Different Support Conditions for $X_{i}$\label{fig:id_set_plot}}
\end{figure}
\par\end{center}

In Figure \ref{fig:id_set_plot}, the black dot represents the true
$\b_{0}$, which is guaranteed to belong to the identified set $B_{0}$.
The red rectangle demonstrates the identified set when the supports
of all coordinates of $X$ are continuous and bounded. The theory
predicts point identification when the number of individuals $N$
and $\left(i,j\right)$ pairs $M$ go to infinity. Here the numerical
result is indeed very close to point identification. The blue rectangle
shows the identified set $B_{0}$ when the first dimension of $X$
is binary while all the other coordinates are continuous and bounded.
The size of $B_{0}$ is larger than in case (1), which is expected
due to the discreteness of $X_{i,1}$. The magenta rectangle corresponds
to the case when $X_{i,1}$ is binary, $X_{i,2}$ is discrete, and
$X_{i,3}$ is continuous and bounded. The size of $B_{0}$ is reasonably
small, given that there are two discrete variables in $X$ and one
of them is binary. The green and black rectangles illustrate $B_{0}$
when all coordinates of $X$ are discretely distributed with equal
probability on 101 (fine grid) and 11 (coarse grid) points between
-0.5 and 0.5, respectively. In these two scenarios, we see larger
identified sets than in (1)--(3). That said, the sets still appear
small relative to the full parameter space $\left[-\pi/2,\pi/2\right]\times\left[-\pi,\pi\right]$.
To summarize, the size of the identified set $B_{0}$ based on logical
differencing is reasonably small under each of the five settings.

\subsection{\label{subsec:Setup-of-Simulation} Performance of the Tetrad Estimator}

We maintain \eqref{eq:simrule} as our network formation model. We
draw each coordinate of $X_{i}$ independently from a uniform distribution
on $\left[-0.5,0.5\right]$ and calculate $w$ using coordinate-wise
absolute difference. Point identification is guaranteed since Assumption
\ref{as3prime} is satisfied. We set $A_{i}=corr\times X_{i,1}+\left(1-corr\right)\times\xi_{i}$,
and use the same distributions as in Section \ref{subsec:Identified-Set}
to generate $\xi_{_{i}}$ and $\epsilon_{ij}$. We set the true $\b_{0}$
to be $(1,...,1)^{'}/\sqrt{d}\in\S^{d-1}$ , and estimate the direction
of $\b_{0}$.

For the baseline result in the main text, we fix the number of individuals
$N=100$, the dimension of $X$ (also $W_{ij}:=w\left(X_{i},X_{j}\right)$
and $\b_{0}$) $d_{x}=d=3$, and $corr=0.2$. In Appendix \ref{subsec:Sim_Robustness},
we vary $N$, $d$ and $corr$, as well as allow for asymmetric $w$
function to investigate how robust our method is against various configurations.

To summarize, for each of the $B=100$ simulations we randomly generate
data on the characteristics of and the network structure among individuals.
Then, based on the observable $\left(X_{i},W_{ij},D_{ij}\right)_{i,j\in\left\{ 1,...,N\right\} }$
matrix we construct our two-step estimator $\widehat{\b}$ for the
true parameter of interest $\beta_{0}$. Specifically, we use a sieve
estimator with 2nd-order spline with its knot at median for the first-stage
nonparametric estimation of $\rho_{i}\left(\cd\right)$. The spline
is chosen to ensure a relatively small number of regressors in the
nonparametric regression considering the small size of $N$. In the
second stage, we adapt to the adaptive-gird search on the unit sphere
algorithm developed in \citet{gao2018robust} to calculate $\widehat{\b}$
that minimizes the sample criterion function $\widehat{Q}\left(\b\right)$
defined in \eqref{eq:sampeQ}\footnote{We suppress its dependence on $\gamma$ and $N$ for notational simplicity.
In all simulations and empirical application, we use smoothing function
$\g\left(t\right):=2\times\Phi\left(\left[t\right]_{+}\right)-1$.} over the unit sphere. It should be noted that, constrained by computational
power, when calculating the sample criterion $\widehat{Q}\text{\ensuremath{\left(\b\right)}}$
for each $\b\in\S^{d-1}$ we randomly draw $M=1,000$ $\left(i,j\right)$
pairs of individuals and vary across all possible $\left(k,l\right)$
pairs excluding $i$ or $j$. One can improve those results by increasing
$M$ when computational constraint is not present, so again our results
are conservative. Lastly, we compare our estimator $\widehat{\b}$
with $\b_{0}$ based on several performance metrics including root
mean squared error (rMSE), mean norm deviations (MND), and maximum
mean absolute deviation (MMAD).

We define for each simulation round $b$ the set estimator $\widehat{\Theta}_{b}$
as the set of points that achieve the minimum of $\widehat{Q}\text{\ensuremath{\left(\b\right)}}$
over the unit sphere $\S^{d-1}$. We further define for each simulation
$b=1,...,B$ and each dimension $h=1,...,d$ of $\beta$ 
\[
\widehat{\b}_{b,h}^{l}:=\min\widehat{\Theta}_{b,h},\quad\widehat{\b}_{b,h}^{u}:=\max\widehat{\Theta}_{b,h},\quad\widehat{\b}_{b,h}^{m}:=\dfrac{1}{2}\left(\widehat{\b}_{b,h}^{l}+\widehat{\b}_{b,h}^{u}\right),
\]
where $\widehat{\b}_{b,h}^{l}$, $\widehat{\beta}_{b,h}^{u}$ and
$\widehat{\b}_{b,h}^{m}$ are the minimum, maximum, and middle point
along dimension $h$ for simulation round $b$ of the identified set
$\widehat{\Theta}_{b}$, respectively. One can consider $\widehat{\b}^{m}$
as the point estimator for $\beta_{0}$. Note by construction for
each simulation round $b$, the identified set $\widehat{\Theta}_{b}$
is a subset of the rectangle 
\[
\widehat{\Xi}_{b}:=\times_{h=1}^{d}\left[\widehat{\b}_{b,h}^{l},\ \widehat{\b}_{b,h}^{u}\right].
\]
We calculate the baseline performance using $\widehat{\b}^{l},\widehat{\b}^{u},\widehat{\b}^{m}$
respectively.

We report in Table \ref{tab:C5_Sim_sum} the performance of our estimators.
In the first row, we calculate the mean bias across $B=100$ simulations
using $\widehat{\b}^{m}$ along each dimension $h=1,...,d$. The result
shows the estimation bias is very small across all dimensions with
a magnitude between -0.0053 and 0.0052. Similar performance is observed
using $\widehat{\b}^{u}$ and $\widehat{\b}^{l}$ as shown in row
2 and 3. We do not find any sign of persistent over/under- estimation
of $\beta_{0}$ across each dimension. Row 4 measures the average
width along each dimension of the rectangle $\widehat{\Xi}_{b}$ over
$B$ simulations. The average size of $\widehat{\Xi}_{b}$ is small,
indicating a tight area for the estimated set. In the second part
of Table \ref{tab:C5_Sim_sum} we report rMSE, MND, and MMAD, all
of which are small in magnitude and provide evidence that our estimator
works well in finite sample.

\begin{table}
\caption{\label{tab:C5_Sim_sum}Baseline Performance}

\bigskip{}

\noindent \centering{}%
\begin{tabular}{ccccc}
\hline 
$\phantom{\frac{\frac{1}{1}}{\frac{1}{1}}}$ &  & $\b_{1}$ & $\b_{2}$ & $\b_{3}$\tabularnewline
\hline 
$\phantom{\frac{\frac{1}{1}}{\frac{1}{1}}}$bias & $\frac{1}{B}\sum_{b}\left(\hat{\b}_{b,h}^{m}-\b_{0,h}\right)$ & -0.0021 & 0.0052 & -0.0053\tabularnewline
$\phantom{\frac{\frac{1}{1}}{\frac{1}{1}}}$upper bias & $\frac{1}{B}\sum_{b}\left(\hat{\b}_{b,h}^{u}-\b_{0,h}\right)$ & 0.0048 & 0.0118 & -0.0002\tabularnewline
$\phantom{\frac{\frac{1}{1}}{\frac{1}{1}}}$lower bias & $\frac{1}{B}\sum_{b}\left(\hat{\b}_{b,h}^{l}-\b_{0,h}\right)$ & -0.0091 & -0.0015 & -0.0105\tabularnewline
$\phantom{\frac{\frac{1}{1}}{\frac{1}{1}}}$mean$(u-l)$ & $\frac{1}{B}\sum_{b}\left(\hat{\b}_{b,h}^{u}-\hat{\b}_{b,h}^{l}\right)$ & 0.0138 & 0.0132 & 0.0103\tabularnewline
\hline 
 &  &  &  & \tabularnewline
\hline 
$\phantom{\frac{\frac{1}{1}}{\frac{1}{1}}}$rMSE & $\sqrt{\frac{1}{B}\sum_{b}\norm{\hat{\b}_{b}^{m}-\beta_{0}}^{2}}$ & \multicolumn{3}{c}{0.0488}\tabularnewline
$\phantom{\frac{\frac{1}{1}}{\frac{1}{1}}}$MND & $\frac{1}{B}\sum_{b}\norm{\hat{\b}_{b}^{m}-\beta_{0}}$ & \multicolumn{3}{c}{0.0417}\tabularnewline
$\phantom{\frac{\frac{1}{1}}{\frac{1}{1}}}$MMAD & $\max_{h\in\left\{ 1,..,d\right\} }\abs{\frac{1}{B}\sum_{b}\left(\hat{\b}_{b,h}^{m}-\b_{0,h}\right)}$ & \multicolumn{3}{c}{0.0053}\tabularnewline
\hline 
\end{tabular}
\end{table}

\section{\label{sec:Emp}Empirical Illustration}

As an empirical illustration, we estimate a network formation model
under NTU with data of a village network called Nyakatoke in Tanzania.
Nyakatoke is a small Haya community of 119 households in 2000 located
in the Kagera Region of Tanzania. We are interested in how important
factors, such as wealth, distance, and blood or religious ties, are
relative to each other in deciding the formation of risk-sharing links
among local residents. We apply our two-stage estimator to the Nyakatoke
network data and obtain economically intuitive results.

\subsection{\label{subsec:Data}Data Description}

The risk-sharing data of Nyakatoke, collected by Joachim De Weerdt
in 2000, cover all of the 119 households in the community. It includes
the information about whether or not two households are linked in
the insurance network. It also provides detailed information on total
USD assets and religion of each household, as well as kinship and
distance between households. See \citet{de2004risk}, \citet{DeWeerdt2006},
and \citet{de2011social} for more details of this dataset.

To define the dependent variable \textit{link}, the interviewer asks
each household the following question:

$\ $

``\textit{Can you give a list of people from inside or outside of
Nyakatoke, who you can personally rely on for help and/or that can
rely on you for help in cash, kind or labor?}'' 

$\ $

The data contains three answers of ``bilaterally mentioned'', ``unilaterally
mentioned'', and ``not mentioned'' between each pair of households.
Considering the question is about whether one can rely on the other
for help, we interpret both ``bilaterally mentioned'' and ``unilaterally
mentioned'' as they are connected in this undirected network, meaning
that the dependent variable $D_{ij}$ \textit{link} equals 1.\footnote{In the context of the village economies in our application, we think,
at the time of link formation, the risk-sharing links are less likely
(in comparison with the contexts of business or financial networks)
to be driven by efficient arrangements of side-payment transfers,
thus satisfying NTU.} We also ran a robustness check by constructing a weighted network
based on the answers, i.e. ``bilaterally mentioned'' means \textit{link}
equals 2, ``unilaterally mentioned'' means \textit{link} equals
1, and ``not mentioned'' means \textit{link} equals 0, and obtained
very similar results.

We estimate the coefficients for 3 regressors: \emph{wealth difference},
\emph{distance} and \emph{tie} between households, with our two-step
estimator. \textit{Wealth} is defined as the total assets in USD owned
by each household in 2000, including livestocks, durables and land.
\textit{Distance} measures how far away two households are located
in kilometers. \textit{Tie} is a discrete variable, defined to be
3 if members of one household are parents, children and/or siblings
of members of the other household, 2 if nephews, nieces, aunts, cousins,
grandparents and grandchildren, 1 if any other blood relation applies
or if two households share the same religion, and 0 if no blood religious
tie exists. Following the literature we take natural log on \textit{wealth}
and \textit{distance}, and we construct the \textit{wealth difference}
variable as the absolute difference in \emph{wealth}, i.e.,
\[
w\left(X_{i},X_{j}\right)=\left(\begin{array}{c}
\abs{\ln\text{wealth}_{i}-\ln\text{wealth}_{j}}\\
\ln\text{distance}_{ij}\\
\text{tie}_{ij}
\end{array}\right).
\]

Figure \ref{fig:NetworkGraph} illustrates the structure of the insurance
network in Nyakatoke. Each node in the graph represents a household.
The solid line between two nodes indicates they are connected, i.e.,
\textit{link} equals 1. The size of each node is proportional to the
USD wealth of each household. Each node is colored according to its
rank in \emph{wealth}: green for the top quartile, red for the second,
yellow for the third and purple for the fourth quartile.
\begin{center}
\begin{figure}[h]
\begin{centering}
\includegraphics[scale=0.45]{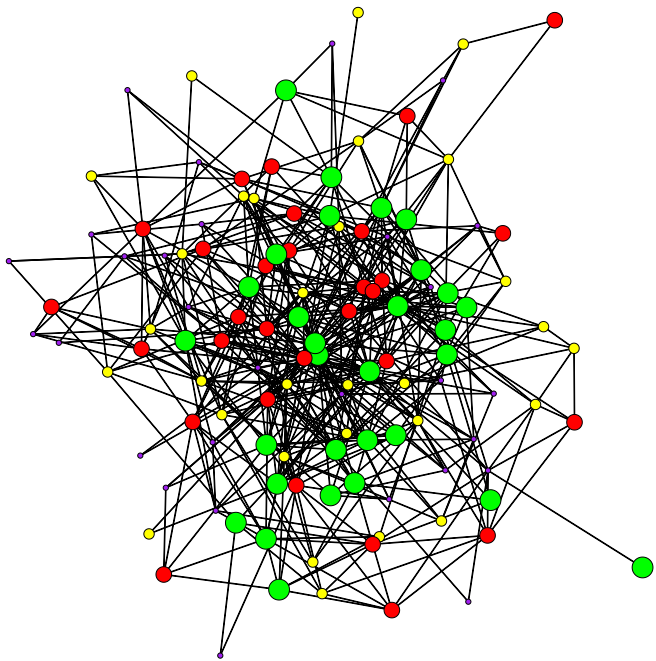}
\par\end{centering}
\caption{A Graphical Illustration of the Insurance Network of Nyakatoke\label{fig:NetworkGraph}}
\end{figure}
\par\end{center}

In the dataset there are 5 households that lack information on \textit{wealth}
and/or \textit{distance}. We drop these observations, resulting in
a sample size $N$ of 114. The total number of ordered household pairs
is 12,882. Summary statistics for the dependent and explanatory variables
used in our analysis are presented in Table \ref{tab:Empirical-Application:-Summary}.

\begin{table}
\caption{Empirical Application: Summary Statistics\label{tab:Empirical-Application:-Summary}}

\bigskip{}

\noindent \centering{}%
\begin{tabular}{cccccc}
\hline 
Variable & Obs & Mean & Std. Dev. & Min & Max\tabularnewline
\hline 
 &  &  &  &  & \tabularnewline
link & 12,882 & 0.0732 & 0.2606 & 0 & 1\tabularnewline
$\abs{\text{(ln) wealth difference}}$ & 12,882 & 1.0365 & 0.8228 & 0.0004 & 5.8898\tabularnewline
($\ln$) distance & 12,882 & 6.0553 & 0.7092 & 2.6672 & 7.4603\tabularnewline
tie & 12,882 & 0.4260 & 0.6123 & 0.0000 & 3.0000\tabularnewline
\hline 
\end{tabular}
\end{table}

Given the data structure, we should expect close to point identification
because even though the \textit{tie} variable is discrete, the other
two regressors \textit{wealth} and \textit{distance} can be considered
as being continuously distributed with a large conditional support.
In addition, \textit{tie} also has more than one point in its support
given other variables. Thus, it is straightforward to verify that
Assumption \textit{\emph{\ref{as4prime2}}}\emph{ }\textit{\emph{is
satisfied, leading to point identification of $\b_{0}$.}}

\subsection{\label{subsec: emp Results-and-Discussion}Results and Discussion}

We apply our two-stage estimator proposed in Section \ref{subsec:DNF_est}
to the Nyakatoke network data, where, similarly to the simulation
studies, we use second degree splines to estimate $\rho_{i}$ and
the adaptive grid search algorithm to find the minimizer of the criterion
function $\widehat{Q}\left(\b\right)$ on $\S^{d-1}$.

\begin{table}[h]
\caption{Empirical Application: Estimation Results\label{tab:EmpResults}}

\bigskip{}

\noindent \centering{}%
\begin{tabular}{ccc}
\hline 
\multirow{2}{*}{\textcolor{black}{$\phantom{\frac{\frac{1}{1}}{\frac{1}{1}}}$}Variable} & \multirow{2}{*}{$\hat{\b}^{m}$} & \multirow{2}{*}{$\left[\hat{\b}^{l},\ \hat{\b}^{u}\right]$}\tabularnewline
 &  & \tabularnewline
\hline 
 &  & \tabularnewline
\textcolor{black}{$\phantom{\frac{\frac{1}{1}}{\frac{1}{1}}}$}$\abs{\text{(ln) wealth difference}}$ & -0.1948 & $\left[-0.1964,\ -0.1932\right]$\tabularnewline
\textcolor{black}{$\phantom{\frac{\frac{1}{1}}{\frac{1}{1}}}$}($\ln$)
distance & -0.8036 & $\left[-0.8043,\ -0.8029\right]$\tabularnewline
\textcolor{black}{$\phantom{\frac{\frac{1}{1}}{\frac{1}{1}}}$}tie & 0.5619 & $\left[0.5608,\ 0.5630\right]$\tabularnewline
\hline 
\end{tabular}
\end{table}

Table \ref{tab:EmpResults} summarizes our estimation results. The
column of $\hat{\b}^{m}$ corresponds to the center of the estimated
rectangle 
\[
\widehat{\Xi}:=\times_{h=1}^{d}\left[\widehat{\b}_{h}^{l},\ \widehat{\b}_{h}^{u}\right],
\]
where $\widehat{\b}_{h}^{l}$ and $\widehat{\b}_{h}^{u}$ represent
the lower and upper bound of the estimated area (the set of minimizers
of sample criterion function \eqref{eq:sampeQ}) along each dimension
$h=1,..,d$. We use $\hat{\b}^{m}$ as the point estimator of the
coefficients. While the scale of $\b_{0}$ is unidentified, one can
still compare the estimated coefficients with each other to obtain
an idea about which variable affects the formation of the link more
than the other.

The estimated coefficients for each variable conform well with economic
intuition. The estimated set for\textit{ wealth difference}'s coefficient
is $\left[-0.1964,-0.1932\right]$, which implies the more difference
in wealth between two households, the lower likelihood they are connected.
The estimated set for \textit{distance} is $\left[-0.8043,-0.8029\right]$.
It is natural that households rely more on neighbors for help than
those who live farther away, thus neighbors are more likely to be
connected. The estimated coefficient for \textit{tie} falls in the
positive range of $\left[0.5608,0.5630\right]$, consistent with the
intuition that one would depend on support from family when negative
shock occurs.

It is worth mentioning the estimated set $\hat{\Xi}$ is very tight
in each dimension, with the maximum width $\max_{h=1,..,d}\left(\hat{\b}_{h}^{u}-\hat{\b}_{h}^{l}\right)$
equal to 0.0032 for \textit{wealth difference.} Usually the discreteness
in \textit{tie} could make the estimated set wide, but our method
is able to leverage the large support in \textit{wealth difference}
and \textit{distance}. Once again, the relative magnitude and sign
of the coefficient for \textit{tie} are estimated in line with expectation
despite of its discreteness. In summary, the empirical results show
that our proposed estimator is able to generate economically intuitive
estimates under NTU.

\section{\label{sec:C5_conc} Conclusion}

This paper considers a semiparametric model of dyadic network formation
under nontransferable utilities, a natural and realistic micro-theoretical
feature that translates into the lack of additive separability in
econometric modeling. We show how a new methodology called \emph{logical
differencing} can be leveraged to cancel out the two-way fixed effects,
which correspond to unobserved individual heterogeneity, without relying
on arithmetic additivity. The key idea is to exploit the logical implication
of weak multivariate monotonicity and use the intersection of mutually
exclusive events on the unobserved fixed effects. It would be interesting
to explore whether and how the idea of \emph{logical differencing},
or more generally the use of fundamental logical operations, can be
applied to other econometric settings.

The identified sets derived by our method under various support restrictions
demonstrate that the proposed method is able to achieve accurate identified
set despite of discreteness and boundedness in $Supp\left(X\right)$.
Simulation results show that our method performs reasonably well with
a relatively small sample size, and robust to various configurations.
The empirical illustration using the real network data of Nyakatoke
reveals that our method is able to capture the essence of the network
formation process by generating estimates that conform well with economic
intuition.

This paper also reveals several further research questions regarding
dyadic network formation models under the NTU setting. First, given
the observation that the NTU setting can capture ``homophily effects''
with respect to the unobserved heterogeneity (under log-concave error
distributions) while imposing monotonicity in the unobserved heterogeneity,
it is interesting to investigate whether one can differentiate homophily
effects generated by ``intrinsic preference'' from assortativity
effects generated by bilateral consent, NTU and log-concave errors.
Second, admittedly the identifying restriction obtained in this paper
becomes uninformative when we have \emph{antisymmetric} pairwise observable
characteristics. However, preliminary analysis based on an adaption
of \citet{gao2018nonparametric} to the NTU setting suggests that
individual unobserved heterogeneity can be nonparametrically identified
up to location and inter-quantile range normalizations. After the
identification of individual unobserved heterogeneity terms ($A_{i}$),
it becomes straightforward to identify the index parameter $\b_{0}$
based on the observable characteristics, even in the presence of antisymmetric
pairwise characteristics. However, consistent estimators of $A_{i}$
and $\b_{0}$ in a semiparametric framework based on the identification
strategy in \citet{gao2018nonparametric} are still being developed.
We thus leave these research questions to future work.

\paragraph{Acknowledgements}

Wayne Gao is grateful to Xiaohong Chen and Peter C. B. Phillips for
their advice and support while working on this paper during his PhD
study at Yale University. Ming Li is grateful to Donald W. K. Andrews
and Yuichi Kitamura for their advice and support while working on
this paper during his PhD study at Yale University. We also thank
Isaiah Andrews, Yann Bramoull\'{e}, Benjamin Connault, and Paul Goldsmith-Pinkham,
as well as seminar and conference participants at Yale University,
University of Arizona, University of Essex, the International Conference
for Game Theory at Stony Brook University (2019), the Young Economist
Symposium at Columbia University (2019), and the Latin American Meeting
of the Econometric Society at Benem\'{e}rita Universidad Aut\'{o}noma
de Puebla (2019) for their comments. All errors are our own.

\bibliographystyle{ecta}
\bibliography{NetForm}

\newpage{}

\appendix

\section*{Appendix}

\section{Proofs\label{sec:Proofs}}

\subsection{\label{subsec:Proof-of-Lemma-2}Proof of Lemma \ref{lem:DNF2}}
\begin{proof}
For notational simplicity, we denote $\D\left(x;x_{i},x_{j}\right)$
to be $w\left(x_{i},x\right)-w\left(x_{j},x\right)$ and write

\begin{equation}
\l\left(\ol x,\ul x;x_{i},x_{j};\b\right)=\ind\left\{ \D\left(\ol x;x_{i},x_{j}\right)^{'}\b\leq0\right\} \ind\left\{ \D\left(\ul x;x_{i},x_{j}\right)^{'}\b\geq0\right\} .\label{eq:def lambda_ij}
\end{equation}
Therefore, the event \eqref{eq:DNF_ID2} is equivalent to $\left\{ \D\left(\ol x;x_{i},x_{j}\right)^{'}\b_{0}>0\right\} \cup\left\{ \D\left(\ul x;x_{i},x_{j}\right)^{'}\b_{0}<0\right\} $
and the event \eqref{eq:DNF_ID3} is equivalent to $\left\{ \D\left(\ol x;x_{i},x_{j}\right)^{'}\b\leq0\right\} \cap\left\{ \D\left(\ul x;x_{i},x_{j}\right)^{'}\b\geq0\right\} $. 

By Assumption \ref{assu:DNF-supportX}, there exist $x_{i}$ and $x_{j}$
in $Supp\left(X_{i}\right)$ such that $\D\left(X_{k};x_{i},x_{j}\right)$
contains ${\bf 0}$ as an interior point, or in other words, contains
all directions from the origin. Hence, given any directions $\b_{0}$
and $\b\neq\b_{0}$ in $\S^{d-1}$, there must exist some $\ol x\in Supp\left(X_{i}\right)$
such that
\begin{equation}
\D\left(\ol x;x_{i},x_{j}\right){}^{'}\b_{0}>0\text{ and }\D\left(\ol x;x_{i},x_{j}\right){}^{'}\text{\ensuremath{\b}}<0,\label{eq:ineq2_u}
\end{equation}
and some $\ul x\in Supp\left(X_{i}\right)$ such that
\begin{equation}
\D\left(\ul x;x_{i},x_{j}\right){}^{'}\b_{0}<0\text{ and }\D\left(\ul x;x_{i},x_{j}\right){}^{'}\text{\ensuremath{\b}}>0.\label{eq:ineq4_l}
\end{equation}
Since all the inequalities above are strict and $w$ is continuous
(Assumption \ref{assu:DNF-con}), there exists some $\e>0$ such that,
for every 
\[
\left(\tilde{x}_{i},\tilde{x}_{j},\tilde{\ol x},\tilde{\ul x}\right)\in\Upsilon:=B_{\e}\left(x_{i}\right)\times B_{\e}\left(x_{j}\right)\times B_{\e}\left(\ol x\right)\times B_{\e}\left(\ul x\right),
\]
with $B_{\e}\left(x\right)$ denoting the open ball of radius $\e$
around $x$, we have
\[
\D\left(\tilde{\ol x};\tilde{x}_{i},\tilde{x}_{j}\right){}^{'}\b_{0}>0\text{ and }\D\left(\tilde{\ol x};\tilde{x}_{i},\tilde{x}_{j}\right){}^{'}\text{\ensuremath{\b}}<0,
\]
and
\[
\D\left(\tilde{\ul x};\tilde{x}_{i},\tilde{x}_{j}\right){}^{'}\b_{0}<0\text{ and }\D\left(\tilde{\ul x};\tilde{x}_{i},\tilde{x}_{j}\right){}^{'}\text{\ensuremath{\b}}>0,
\]
which imply \eqref{eq:DNF_ID2} and \eqref{eq:DNF_ID3}. Since $x_{i},x_{j},\ol x,\ul x$
all belong to $Supp\left(X_{i}\right)$, we have
\begin{equation}
\P\left(\left(X_{i},X_{j},X_{k},X_{l}\right)\in\Upsilon\right)>0,\label{eq:tetrad_posprob}
\end{equation}
when we randomly sample individuals $\left(i,j,k,l\right)$.

Now, fix any $\ol a\in Supp\left(A_{i}\right)$ and any $\left(\tilde{\ol x},\tilde{\ul x},\tilde{x}_{i},\tilde{x}_{j}\right)\in\Upsilon$.
Since $\phi$ is strictly increasing in its first argument and continuous
in all its arguments, the function $\psi_{\tilde{\ol x}}$
\[
\psi_{\tilde{\ol x}}\left(w\left(\tilde{x}_{i},\tilde{\ol x}\right)^{'}\b_{0},a_{i}\right)=\int\phi\left(w\left(\tilde{x}_{i},\tilde{\ol x}\right)^{'}\b_{0},a_{i},A_{k}\right)d\P\left(\rest{A_{k}}X_{k}=\tilde{\ol x}\right)
\]
is also strictly increasing in its first argument and continuous in
all its arguments. Hence, 
\begin{align*}
 & \psi_{\tilde{\ol x}}\left(w\left(\tilde{x}_{i},\tilde{\ol x}\right)^{'}\b_{0},\ol a\right)-\psi_{\tilde{\ol x}}\left(w\left(\tilde{x}_{j},\tilde{\ol x}\right)^{'}\b_{0},\ol a\right)>0,
\end{align*}
and similarly 
\begin{align*}
 & \psi_{\tilde{\ul x}}\left(w\left(\tilde{x}_{i},\tilde{\ul x}\right)^{'}\b_{0},\ol a\right)-\psi_{\tilde{\ul x}}\left(w\left(\tilde{x}_{j},\tilde{\ul x}\right)^{'}\b_{0},\ol a\right)<0.
\end{align*}
Furthermore, there must exist some $\e>0$ such that, for any $a_{i},a_{j}\in\left[\ol a-\e,\ol a+\e\right]$,
and any individuals $i,j$ with $\left(X_{i},A_{i}\right)=\left(\tilde{x}_{i},a_{i}\right)$
and $\left(X_{j},A_{j}\right)=$$\left(\tilde{x}_{j},a_{j}\right)$,
we have
\[
\rho_{i}\left(\tilde{\ol x}\right)-\rho_{j}\left(\tilde{\ol x}\right)=\psi_{\tilde{\ol x}}\left(w\left(\tilde{x}_{i},\tilde{\ol x}\right)^{'}\b_{0},a_{i}\right)-\psi_{\tilde{\ol x}}\left(w\left(\tilde{x}_{j},\tilde{\ol x}\right)^{'}\b_{0},a_{j}\right)>0
\]
and
\[
\rho_{i}\left(\tilde{\ul x}\right)-\rho_{j}\left(\tilde{\ul x}\right)=\psi_{\tilde{\ul x}}\left(w\left(\tilde{x}_{i},\tilde{\ul x}\right)^{'}\b_{0},a_{i}\right)-\psi_{\tilde{\ul x}}\left(w\left(\tilde{x}_{j},\tilde{\ul x}\right)^{'}\b_{0},a_{j}\right)<0,
\]
which implies \eqref{eq:DNF_ID1}. Since
\[
\ol a\in Supp\left(A_{i}\right)=Supp\left(\rest{A_{i}}X_{i}=\tilde{x}_{i}\right)=Supp\left(\rest{A_{j}}X_{j}=\tilde{x}_{j}\right)
\]
by Assumption \ref{assu:DNF-supportA}, we have
\begin{align}
 & \P\left(\rest{\tau_{ij}\left(\tilde{\ol x},\tilde{\ul x}\right)=1}X_{i}=\tilde{x}_{i},X_{j}=\tilde{x}_{j}\right)\nonumber \\
>\  & \P\left(\rest{A_{i},A_{j}\in\left[\ol a-\e,\ol a+\e\right]}X_{i}=\tilde{x}_{i},X_{j}=\tilde{x}_{j}\right)>0.\label{eq:ijA_posprob}
\end{align}

Now, combining \eqref{eq:tetrad_posprob} and \eqref{eq:ijA_posprob},
we have

\begin{align*}
 & \P\left\{ \eqref{eq:DNF_ID1},\eqref{eq:DNF_ID2}\text{ and }\eqref{eq:DNF_ID3}\text{ hold}\right\} \\
\geq\  & \int_{\Upsilon}\P\left(\rest{\tau_{ij}\left(X_{k},X_{l}\right)=1}X_{i}=\tilde{x}_{i},X_{j}=\tilde{x}_{j},X_{k}=\tilde{\ol x},X_{l}=\tilde{\ul x}\right)d\P\left(x_{i},x_{j},\ol x,\ul x\right)>0
\end{align*}
since the integrand is strictly positive on the set $\Upsilon$, which
has strictly positive probability measure under $\P$.
\end{proof}

\subsection{\label{subsec:Proof-of-Theorem 1} Proof of Theorem \ref{thm:DNF-ID}}
\begin{proof}
By Lemma \ref{lem:DNF1}, we have $\b_{0}\in\arg\min_{\b\in\text{\ensuremath{\mathbb{S}}}^{d-1}}Q(\b)$
because $Q\left(\b_{0}\right)=0\leq Q\left(\b\right)$ by the construction
of the population criterion $Q\left(\cd\right)$. Furthermore, we
have $\b_{0}$ is the unique minimizer of $Q(\b)$ because for any
$\b\neq\b_{0}$, we have
\begin{align}
Q\left(\b\right) & =\E\left[\l\left(X_{k},X_{l};X_{i},X_{j};\b\right)\tau_{ij}\left(X_{k},X_{l}\right)\right]\nonumber \\
 & \geq\P\left\{ \eqref{eq:DNF_ID1},\eqref{eq:DNF_ID2}\text{ and }\eqref{eq:DNF_ID3}\text{ hold}\right\} >0,\label{eq:unique minimizer}
\end{align}
by Lemma \ref{lem:DNF2}.

Now, for any $\b\in\S^{d-1}$, the function $g_{ij}\left(z,\b\right):=\l\left(\ol x,\ul x;x_{i},x_{j};\b\right)\tau_{ij}\left(\ol x,\ul x\right)$
where $z:=\left(\ol x,\ul x;x_{i},x_{j}\right)$ is discontinuous
in $\b$ only on the (finite union of) hyperplanes defined by $\D(\ol x;x_{i},x_{j})^{'}\b=0$
or $\D(\ul x;x_{i},x_{j})^{'}\b=0$, which are probability zero events
under Assumption \ref{assu:DNF-supportX}. Moreover, 
\begin{equation}
\E\sup_{\b\in\S^{d-1}}\abs{g_{ij}\left(z,\b\right)}\leq1<\infty.\label{eq:part (ii) newey}
\end{equation}
Hence, by \citet*{newey1994asymp}, $Q\left(\b\right)$ is continuous,
and, given that $\S^{d-1}$ is compact, the desired result in Theorem
\ref{thm:DNF-ID} follows.
\end{proof}

\subsection{\label{subsec:Proof-of-Lemma 3}Lemma \ref{lem:DNF-UC} with Proof}

We state and prove the following lemma that we use to prove Theorem
\ref{thm:DNF-consistent}.
\begin{lem}[Uniform Convergence of $\widehat{Q}_{n}^{\g}\left(\b\right)$]
\label{lem:DNF-UC}Under model \eqref{eq:Model_NF} and Assumptions
\ref{as1prime}, \ref{assu:DNF-rs}--\ref{assu:as_gamma}, we have
\[
\sup_{\b\in\text{\ensuremath{\mathbb{S}}}^{d-1}}\abs{\widehat{Q}_{n}^{\g}\left(\b\right)-Q^{\g}\left(\b\right)}\pto0.
\]
\end{lem}
\begin{proof}
For notational simplicity, we suppress the superscript $\g$ in this
proof. Define the infeasible criterion $\widetilde{Q}_{n}\left(\b\right)$
as

\begin{equation}
\begin{aligned}\widetilde{Q}_{n}\left(\b\right):=\frac{\left(n-4\right)!}{n!}\sum_{1\leq i\neq j\neq k\neq l\leq n} & \tau_{ij}^{\g}\left(X_{k},X_{l}\right)\l\left(X_{k},X_{l};X_{i},X_{j};\b\right).\end{aligned}
\label{eq:infeasible Q}
\end{equation}
By triangular inequality, we have

\begin{equation}
\sup_{\b\in\text{\ensuremath{\mathbb{S}}}^{d-1}}\abs{\widehat{Q}_{n}\left(\b\right)-Q\left(\b\right)}\leq\sup_{\b\in\text{\ensuremath{\mathbb{S}}}^{d-1}}\abs{\widehat{Q}_{n}\left(\b\right)-\widetilde{Q}\left(\b\right)}+\sup_{\b\in\text{\ensuremath{\mathbb{S}}}^{d-1}}\abs{\widetilde{Q}\left(\b\right)-Q\left(\b\right)}.\label{eq:triangle ienquality}
\end{equation}

\medskip{}

We first show that $\sup_{\b\in\text{\ensuremath{\mathbb{S}}}^{d-1}}\abs{\widehat{Q}_{n}\left(\b\right)-\widetilde{Q}\left(\b\right)}=o_{p}\left(1\right)$.
Since $\l$ only takes value in $\left\{ 0,1\right\} $, we have
\begin{equation}
\begin{aligned} & \sup_{\b\in\text{\ensuremath{\mathbb{S}}}^{d-1}}\abs{\widehat{Q}_{n}\left(\b\right)-\widetilde{Q}\left(\b\right)}\\
= & \ \frac{\left(n-4\right)!}{n!}\sum_{1\leq i\neq j\neq k\neq l\leq n}\sup_{\b\in\text{\ensuremath{\mathbb{S}}}^{d-1}}\abs{\l\left(X_{k},X_{l};X_{i},X_{j};\b\right)}\\
 & \ \times\abs{\g\left(\hat{\rho}_{i}(X_{k})-\hat{\rho}_{j}(X_{k})\right)\cdot\g\left(\hat{\rho}_{j}(X_{l})-\hat{\rho}_{i}(X_{l})\right)-\g\left(\rho_{i}(X_{k})-\rho_{j}(X_{k})\right)\cdot\g\left(\rho_{j}(X_{l})-\rho_{i}(X_{l})\right)}\\
\leq & \ \frac{\left(n-4\right)!}{n!}\sum_{1\leq i\neq j\neq k\neq l\leq n}\abs{\begin{array}{c}
\g\left(\hat{\rho}_{i}(X_{k})-\hat{\rho}_{j}(X_{k})\right)\cdot\g\left(\hat{\rho}_{j}(X_{l})-\hat{\rho}_{i}(X_{l})\right)\\
-\g\left(\rho_{i}(X_{k})-\rho_{j}(X_{k})\right)\cdot\g\left(\rho_{j}(X_{l})-\rho_{i}(X_{l})\right)
\end{array}}\\
\leq & \ \frac{\left(n-4\right)!}{n!}\sum_{1\leq i\neq j\neq k\neq l\leq n}\left(\begin{array}{c}
\left|\g\left(\hat{\rho}_{i}(X_{k})-\hat{\rho}_{j}(X_{k})\right)-\g\left(\rho_{i}(X_{k})-\rho_{j}(X_{k})\right)\right|\cd\g\left(\rho_{j}(X_{l})-\rho_{i}(X_{l})\right)\\
+\g\left(\rho_{i}(X_{k})-\rho_{j}(X_{k})\right)\cdot\left|\g\left(\hat{\rho}_{j}(X_{l})-\hat{\rho}_{i}(X_{l})\right)-\g\left(\rho_{j}(X_{l})-\rho_{i}(X_{l})\right)\right|\\
+\left|\g\left(\hat{\rho}_{i}(X_{k})-\hat{\rho}_{j}(X_{k})\right)-\g\left(\rho_{i}(X_{k})-\rho_{j}(X_{k})\right)\right|\\
\cd\left|\g\left(\hat{\rho}_{j}(X_{l})-\hat{\rho}_{i}(X_{l})\right)-\g\left(\rho_{j}(X_{l})-\rho_{i}(X_{l})\right)\right|
\end{array}\right)\\
\leq & \ C\cd\frac{\left(n-4\right)!}{n!}\sum_{1\leq i\neq j\neq k\neq l\leq n}\left(\begin{array}{c}
\left(\left|\hat{\rho}_{i}(X_{k})-\rho_{i}(X_{k})\right|+\left|\hat{\rho}_{j}(X_{k})-\rho_{j}(X_{k})\right|\right)\cd\left|\rho_{j}(X_{l})-\rho_{i}(X_{l})\right|\\
+\left(\left|\hat{\rho}_{i}(X_{l})+\rho_{i}(X_{l})\right|+\left|\hat{\rho}_{j}(X_{l})-\rho_{j}(X_{l})\right|\right)\cd\left|\rho_{i}(X_{k})-\rho_{j}(X_{k})\right|\\
+\left(\left|\hat{\rho}_{i}(X_{k})-\rho_{i}(X_{k})\right|+\left|\hat{\rho}_{j}(X_{k})-\rho_{j}(X_{k})\right|\right)\\
\cd\left(\left|\hat{\rho}_{i}(X_{l})-\rho_{i}(X_{l})\right|+\left|\hat{\rho}_{j}(X_{l})-\rho_{j}(X_{l})\right|\right)
\end{array}\right),
\end{aligned}
\label{eq:splitdiff}
\end{equation}
for some $C>0$, where the last inequality holds due to Lipschitz-continuity
of $\g$ in Assumption \ref{assu:as_gamma}.

Recall that $\rho_{i}$ as a function is identified by $\left(x_{i},a_{i}\right)$
through
\[
\rho_{i}\left(x\right)=\int\phi\left(w\left(x,x_{i}\right)^{'}\b_{0},a_{i},A_{k}\right)\P\left(\rest{A_{k}}X_{k}=x\right).
\]
We write the $L_{2}\left(\P_{X}\right)$ norm of $\rho_{i}$ to mean
\[
\norm{\rho_{i}}:=\sqrt{\int\rho_{i}^{2}\left(x\right)d\P\left(X_{k}=x\right)}
\]
and use the subscripts of $\E$ to denote expectation over variables
indexed by those subscripts, e.g., 
\[
\E_{i}\norm{\rho_{i}}:=\int\norm{\rho_{i}}d\P\left(X_{i}=x_{i},A_{i}=a_{i}\right).
\]
Then, by \eqref{eq:splitdiff}, we have
\begin{align*}
 & \ \E\sup_{\b\in\text{\ensuremath{\mathbb{S}}}^{d-1}}\abs{\widehat{Q}_{n}\left(\b\right)-\widetilde{Q}\left(\b\right)}\cd\frac{1}{C}\\
\leq & \ \E_{ijkl}\left[\left(\left|\hat{\rho}_{i}(X_{k})-\rho_{i}(X_{k})\right|+\left|\hat{\rho}_{j}(X_{k})-\rho_{j}(X_{k})\right|\right)\cd\left|\rho_{j}(X_{l})-\rho_{i}(X_{l})\right|\right]\\
 & \ +\E_{ijkl}\left[\left(\left|\hat{\rho}_{i}(X_{l})-\rho_{i}(X_{l})\right|+\left|\hat{\rho}_{j}(X_{l})-\rho_{j}(X_{l})\right|\right)\cd\left|\rho_{i}(X_{k})-\rho_{j}(X_{k})\right|\right]\\
 & \ +\E_{ijkl}\left[\left(\left|\hat{\rho}_{i}(X_{k})-\rho_{i}(X_{k})\right|+\left|\hat{\rho}_{j}(X_{k})-\rho_{j}(X_{k})\right|\right)\left(\left|\hat{\rho}_{i}(X_{l})+\rho_{i}(X_{l})\right|+\left|\hat{\rho}_{j}(X_{l})-\rho_{j}(X_{l})\right|\right)\right]\\
= & \ 2\E_{ijkl}\left|\hat{\rho}_{i}(X_{k})-\rho_{i}(X_{k})\right|\left|\rho_{j}(X_{l})-\rho_{i}(X_{l})\right|+2\E_{ijkl}\left|\hat{\rho}_{j}(X_{k})-\rho_{j}(X_{k})\right|\left|\rho_{j}(X_{l})-\rho_{i}(X_{l})\right|\\
 & \ +\E_{ik}\left|\hat{\rho}_{i}(X_{k})-\rho_{i}(X_{k})\right|^{2}+\E_{jk}\left|\hat{\rho}_{j}(X_{k})-\rho_{j}(X_{k})\right|^{2}+\E_{ik}\left|\hat{\rho}_{i}(X_{k})-\rho_{i}(X_{k})\right|\E_{jl}\left|\hat{\rho}_{j}(X_{l})-\rho_{j}(X_{l})\right|\\
\leq & \ 2\E_{ik}\left|\hat{\rho}_{i}(X_{k})-\rho_{i}(X_{k})\right|+2\E_{jk}\left|\hat{\rho}_{j}(X_{k})-\rho_{j}(X_{k})\right|\\
 & \ +\E_{ik}\left|\hat{\rho}_{i}(X_{k})-\rho_{i}(X_{k})\right|^{2}+\E_{jk}\left|\hat{\rho}_{j}(X_{k})-\rho_{j}(X_{k})\right|^{2}+\E_{ik}\left|\hat{\rho}_{i}(X_{k})-\rho_{i}(X_{k})\right|\E_{jl}\left|\hat{\rho}_{j}(X_{l})-\rho_{j}(X_{l})\right|\\
\leq & \ 4\E_{i}\norm{\hat{\rho}_{i}-\rho_{i}}+2\E_{i}\norm{\hat{\rho}_{i}-\rho_{i}}^{2}+\E_{i}\norm{\hat{\rho}_{i}-\rho_{i}}\cd\E_{j}\norm{\hat{\rho}_{j}-\rho_{j}}\\
= & \ 4\E_{i}\left[o_{p}\left(1\right)\right]+2\E_{i}\left[\left(o_{p}\left(1\right)\right)^{2}\right]+\left(\E_{i}\left[o_{p}\left(1\right)\right]\right)^{2}\\
= & \ o_{p}\left(1\right)
\end{align*}
where the second last inequality follows from the observation that
$\left|\rho_{j}-\rho_{i}\right|\leq1$, while the last inequality
follows from the Cauchy--Schwarz inequality, and the second last
equality follows from $\norm{\hat{\rho}_{i}-\rho_{i}}=o_{p}\left(1\right)$
by the $L_{2}\left(\P_{X}\right)$-consistency of $\hat{\rho}_{i}$
for each $i$. Finally, by the Markov inequality, we have
\begin{equation}
\sup_{\b\in\text{\ensuremath{\mathbb{S}}}^{d-1}}\abs{\widehat{Q}_{n}\left(\b\right)-\widetilde{Q}\left(\b\right)}=o_{p}\left(1\right).\label{eq: first step result}
\end{equation}
\medskip{}

Next, we show 
\[
\sup_{\b\in\text{\ensuremath{\mathbb{S}}}^{d-1}}\abs{\widetilde{Q}_{n}\left(\b\right)-Q\left(\b\right)}=o_{p}\left(1\right).
\]
Clearly, $\left\{ \widetilde{Q}_{n}\left(\b\right)-Q\left(\b\right):\b\in\mathbb{S}^{d-1}\right\} $
is a centered U-process of order $4$, and we apply the results in
\citet{arcones1993limit} for U-statistic empirical processes.

To start, we know that the collection of halfspaces $x^{'}\b\geq0$
across $\b\in\S^{d-1}$ is a VC class of functions with VC dimension
$d+2$, by Problem 14 in Section 6 of \citet*[VW thereafter]{van1996weak}.
Furthermore, the intersection of VC classes remains VC by Lemma 2.6.17(ii)
of VW. Hence, $\left\{ \l\left(\cd,\b\right):\b\in\S^{d-1}\right\} $
is a VC-subgraph class of functions. Since $\left\{ \l\left(\cd,\b\right):\b\in\S^{d-1}\right\} $
has a constant envelope function 1, which trivially has finite expectation.
By Corollary 3.3 of \citet*{arcones1993limit}, $\left\{ \l\left(\cd,\b\right):\b\in\S^{d-1}\right\} $
is a Glivenko-Cantelli class,

Next, by VW Corollary 2.7.2 that the bracketing number 
\[
\mathscr{N}_{[]}\left(\e,{\cal C}_{M}^{\left\lfloor d/2\right\rfloor +1}\left(Supp\left(X_{i}\right)\right),\norm{\cd}_{1,\P}\right)<\infty
\]
 for every $\e>0$. Hence, by Corollary 3.3 of \citet*{arcones1993limit},
we know that ${\cal C}_{M}^{\left\lfloor d/2\right\rfloor +1}\left(Supp\left(X_{i}\right)\right)$
is also Glivenko-Cantelli.

Now, since the mapping
\[
\left(\rho_{i},\rho_{j},\l\right)\mapsto\g\left(\rho_{i}-\rho_{j}\right)\g\left(\rho_{j}-\rho_{i}\right)\l\left(\cd,\b\right)
\]
 is continuous in all its arguments, by Theorem 3 of \citet*{van2000preservation},
we know that 
\[
\left\{ \g\left(\rho_{i}-\rho_{j}\right)\g\left(\rho_{j}-\rho_{i}\right)\l\left(\cd,\b\right):\rho_{i},\rho_{j}\in{\cal C}_{M}^{\left\lfloor d/2\right\rfloor +1}\left(Supp\left(X_{i}\right)\right),\b\in\S^{d-1}\right\} 
\]
is also Glivenko-Cantelli, i.e.,
\begin{equation}
\sup_{\b\in\text{\ensuremath{\mathbb{S}}}^{d-1}}\abs{\widetilde{Q}_{n}\left(\b\right)-Q\left(\b\right)}\pto0.\label{eq:UC consistency}
\end{equation}

Combining \eqref{eq: first step result} and \eqref{eq:UC consistency},
we have 
\begin{equation}
\sup_{\b\in\text{\ensuremath{\mathbb{S}}}^{d-1}}\abs{\widehat{Q}_{n}\left(\b\right)-Q\left(\b\right)}\pto0.\label{eq:UC}
\end{equation}
\end{proof}

\subsection{\label{subsec:Proof-of-Theorem2}Proof of Theorem \ref{thm:DNF-consistent}}
\begin{proof}
By the standard theory of M-estimation, say, Theorem 5.7 of \citet*{van2000asymptotic},
the conclusion of Theorem \ref{thm:DNF-consistent} follows from Theorem
\ref{thm:DNF-ID} (clean point identification) and Lemma \ref{lem:DNF-UC}
(uniform convergence of sample criterion).
\end{proof}

\section{\label{subsec:PID_Discrete}Point Identification with a Special Covariate}

\noindent \textbf{Assumption 4$''$}\customlabel{as4prime2}{$4''$}\emph{Suppose
that:}
\begin{itemize}
\item[\emph{(i)}] \emph{ $w_{h}\left(\ol x,\ul x\right):=\left|\ol x_{h}-\ul x_{h}\right|$
for every coordinate $h$;}
\item[\emph{(ii)}] \emph{ $\b_{0,1}\neq0$.}
\item[\emph{(iii)}] \emph{ the support of $X_{i,1}$} given all other coordinates $X_{i,-1}$
is the whole real line $\R$.
\item[\emph{(iv)}] \emph{ there exist two distinct values $\ol x_{h},\ul x_{h}\in\R$
for each coordinate $h\in\left\{ 2,...,d\right\} $ such that $\vartimes_{h=2}^{d}\left\{ \ol x_{h},\ul x_{h}\right\} \subseteq Supp\left(X_{i,-1}\right).$}
\end{itemize}
Assumption\emph{ \ref{as4prime2}} is very similar to the corresponding
set of assumptions imposed for maximum-score estimators a la \citet{horowitz1992smoothed},
with the exception that Assumption\emph{ \ref{as4prime2}}(iv) is
stronger than the corresponding condition in \citet{horowitz1992smoothed},
which only requires that the support of $X_{i}$ is not contained
in any proper linear subspace of $\R^{d}$. Nevertheless, we regard
Assumption\emph{ \ref{as4prime2}}(iv) as a very mild requirement:
it essentially requires that the conditional support of each coordinate
of $X_{i,-1}$ does not degenerate to a singleton, at least in some
part of $Supp\left(X_{i,-1}\right)$. Assumption\emph{ \ref{as4prime2}}(iv)
effectively enables us to vary one discrete coordinate while holding
the other coordinates in $X_{i,-1}$ fixed. Even though this condition
is stronger than necessary, it drastically simplifies the proof of
point identification below.

To see the intuition through a more concrete example, suppose $X_{i}$
is two-dimensional and $\b_{0}=\left(1,1\right)^{'}$. Then, to differentiate
$\b_{0}$ from $\b=\left(1,-1\right)^{'}$, we need to find in-support
$\left(X_{i},X_{j},X_{k},X_{l}\right)$ such that both
\begin{align*}
\text{sgn}\left\{ \left(W_{ik}-W_{jk}\right)^{'}\b_{0}\right\}  & \neq\text{sgn}\left\{ \left(W_{ik}-W_{jk}\right)^{'}\b\right\} \text{ and }\\
\text{sgn}\left\{ \left(W_{il}-W_{jl}\right)^{'}\b_{0}\right\}  & \neq\text{sgn}\left\{ \left(W_{il}-W_{jl}\right)^{'}\b\right\} 
\end{align*}
are satisfied in order to have \eqref{eq:DNF_ID1}, \eqref{eq:DNF_ID2}
and \eqref{eq:DNF_ID3} occur simultaneously with strictly positive
probability. But this is not possible if the second (discrete) dimension
of $X$ is the same for all individuals, since $\b_{0}^{\left(1\right)}=\b^{\left(1\right)}=1$
and the only way to flip the sign is to vary the second coordinate
of $X$ for each individual. The general argument will be made clearer
in the proof of Lemma \emph{\ref{lem2prime2}.}

Given Assumption \emph{\ref{as4prime2}}(ii), it is standard to proceed
with a different normalization from the unit sphere $\S^{d-1}$ in
the main text:
\[
\b_{0}\in{\cal B}:=\left\{ 1,-1\right\} \times\R^{d-1}.
\]

\noindent \textbf{Lemma 2$''$}\customlabel{lem2prime2}{$2''$} \emph{Under
model \eqref{eq:Model_NF}, Assumptions \ref{as1prime}, \ref{assu:DNF-rs},
\ref{assu:DNF-con}, \ref{as4prime2}, and \ref{assu:DNF-supportA},
for each $\b\in{\cal B}\backslash\left\{ \b_{0}\right\} $, \eqref{eq:DNF_ID1},
\eqref{eq:DNF_ID2} and \eqref{eq:DNF_ID3} occur simultaneously with
strictly positive probability.}
\begin{proof}
Take any $\b\in{\cal B}\backslash\left\{ \b_{0}\right\} $. We consider
three separate cases:

\emph{Case 1:} $\b_{1}=\b_{0,1}=1$.

In this case, there exists some coordinate $h\neq1$ such that $\b_{h}\neq\b_{0,h}$.
Set
\begin{align*}
\ul x_{-1} & :=\left(\ul x_{2},..,\ul x_{h}\right)\\
\hat{x}_{-1} & :=\ul x_{-1}+\left(\ol x_{h}-\ul x_{h}\right)e_{h}
\end{align*}
where $e_{h}$ denotes the elementary vector with 1 for coordinate
$h$ and $0$ elsewhere. By Assumption \emph{\ref{as4prime2}(iv),}
$\ul x_{-1},\hat{x}_{-1}\in Supp\left(X_{i,-1}\right)$. Moreover
\[
w_{-1}\left(\hat{x}_{-1},\ul x_{-1}\right)-w_{-1}\left(\ul x_{-1},\ul x_{-1}\right)=\left|\ol x_{h}-\ul x_{h}\right|e_{h}
\]
Now, writing
\begin{align*}
z & :=\left|\ol x_{h}-\ul x_{h}\right|\b_{0,h}\in\R,\\
\d & :=\left|\ol x_{h}-\ul x_{h}\right|\left(\b_{0,h}-\b_{h}\right)\neq0,
\end{align*}
we have
\begin{align*}
 & \left(w_{-1}\left(\hat{x}_{-1},\ul x_{-1}\right)-w_{-1}\left(\ul x_{-1},\ul x_{-1}\right)\right)^{'}\b_{-1}=\left|\ol x_{h}-\ul x_{h}\right|\b_{h}=z-\d\\
\neq\  & z=\left|\ol x_{h}-\ul x_{h}\right|\b_{0,h}=\left(w_{-1}\left(\hat{x}_{-1},\ul x_{-1}\right)-w_{-1}\left(\ul x_{-1},\ul x_{-1}\right)\right)^{'}\b_{0,-1}
\end{align*}
and
\begin{align*}
 & \left(w_{-1}\left(\hat{x}_{-1},\hat{x}_{-1}\right)-w_{-1}\left(\ul x_{-1},\hat{x}_{-1}\right)\right)^{'}\b_{-1}=-\left|\ol x_{h}-\ul x_{h}\right|\b_{h}=-z+\d\\
\neq\  & -z=-\left|\ol x_{h}-\ul x_{h}\right|\b_{0,h}=\left(w_{-1}\left(\hat{x}_{-1},\hat{x}_{-1}\right)-w_{-1}\left(\ul x_{-1},\hat{x}_{-1}\right)\right)^{'}\b_{0,-1}
\end{align*}

Now, by Assumption \emph{\ref{as4prime2}(iii),} we can set
\begin{align*}
\hat{x} & :=\left(\left[\frac{1}{2}\d-z\right]_{+},\ \hat{x}_{-1}\right)\in Supp\left(X_{i}\right),\\
\check{x} & :=\left(\left[z-\frac{1}{2}\d\right]_{+},\ \ul x_{-1}\right)\in Supp\left(X_{i}\right),\\
\tilde{x} & :=\left(0,\ \ul x_{-1}\right)\in Supp\left(X_{i}\right),\\
\mathring{x} & :=\left(\left[\frac{1}{2}\d-z\right]_{+}+\left[z-\frac{1}{2}\d\right]_{+},\ \hat{x}_{-1}\right)\in Supp\left(X_{i}\right).
\end{align*}
so that
\begin{align}
\D\left(\tilde{x};\hat{x},\check{x}\right){}^{'}\b_{0} & =\left(w\left(\hat{x},\tilde{x}\right)-w\left(\check{x},\tilde{x}\right)\right)^{'}\b_{0}\nonumber \\
 & =\left|\hat{x}_{1}\right|-\left|\check{x}_{1}\right|+\left|\ol x_{h}-\ul x_{h}\right|\b_{0,h}\nonumber \\
 & =\left(\frac{1}{2}\d-z\right)+z=\frac{1}{2}\d,\nonumber \\
\D\left(\tilde{x};\hat{x},\check{x}\right){}^{'}\b & =\left(w\left(\hat{x},\tilde{x}\right)-w\left(\check{x},\tilde{x}\right)\right)^{'}\b\nonumber \\
 & =\left(\left|\hat{x}_{1}\right|-\left|\check{x}_{1}\right|\right)+\left|\ol x_{h}-\ul x_{h}\right|\b_{h}\nonumber \\
 & =\left(\frac{1}{2}\d-z\right)+\left(z-\d\right)=-\frac{1}{2}\d,\label{eq:disc_ineq4}\\
\D\left(\mathring{x};\hat{x},\check{x}\right){}^{'}\b_{0} & =\left(w\left(\hat{x},\mathring{x}\right)-w\left(\check{x},\mathring{x}\right)\right)^{'}\b_{0}\nonumber \\
 & =\left(\left|\hat{x}_{1}-\mathring{x}_{1}\right|-\left|\check{x}_{1}-\mathring{x}_{1}\right|\right)-\left|\ol x_{h}-\ul x_{h}\right|\b_{0,h}\nonumber \\
 & =z-\frac{1}{2}\d-z=-\frac{1}{2}\d,\nonumber \\
\D\left(\mathring{x};\hat{x},\check{x}\right){}^{'}\b & =\left(w\left(\hat{x},\mathring{x}\right)-w\left(\check{x},\mathring{x}\right)\right)^{'}\b>0\nonumber \\
 & =\left(\left|\hat{x}_{1}-\mathring{x}_{1}\right|-\left|\check{x}_{1}-\mathring{x}_{1}\right|\right)-\left|\ol x_{h}-\ul x_{h}\right|\b_{h}\nonumber \\
 & =z-\frac{1}{2}\d-z+\d=\frac{1}{2}\d,\nonumber 
\end{align}
which exactly correspond to the inequalities \eqref{eq:ineq2_u} and
\eqref{eq:ineq4_l} in the proof of Lemma \ref{lem:DNF2}, both of
the form
\[
\text{sgn}\left(\D\left(\tilde{x};\hat{x},\check{x}\right){}^{'}\b_{0}\right)=\text{sgn}\left(\D\left(\mathring{x};\hat{x},\check{x}\right){}^{'}\b\right)\neq\text{sgn}\left(\D\left(\mathring{x};\hat{x},\check{x}\right){}^{'}\b_{0}\right)=\text{sgn}\left(\D\left(\tilde{x};\hat{x},\check{x}\right){}^{'}\b\right)
\]
Again, since all the inequalities in \eqref{eq:disc_ineq4} are strict,
they continue to hold for points in sufficiently small open balls
around $\hat{x},\check{x},\tilde{x},\mathring{x}$. And since $\hat{x},\check{x},\tilde{x},\mathring{x}$
are all taken from $Supp\left(X_{i}\right)$, any small open balls
around them have strictly positive probability measures. The rest
of the proof is exactly the same as in the proof of Lemma \ref{lem:DNF2}.\medskip{}

\emph{Case 2:} $\b_{1}=\b_{0,1}=-1$.

This case can be handled in the same way as in Case 1, with appropriate
changes of signs.\medskip{}

\emph{Case 3:} $\b_{1}\neq\b_{0,1}$.

In this case, we can take any point $\ul x_{-1}\in Supp\left(X_{i,-1}\right)$
and set
\begin{align*}
\hat{x} & :=\left(\left[\b_{0,1}\right]_{+},\ \ul x_{-1}\right)\in Supp\left(X_{i}\right),\\
\check{x} & :=\left(\left[-\b_{0,1}\right]_{+},\ \ul x_{-1}\right)\in Supp\left(X_{i}\right),\\
\tilde{x} & :=\left(0,\ \ul x_{-1}\right)\in Supp\left(X_{i}\right),\\
\mathring{x} & :=\left(\left[\b_{0,1}\right]_{+}+\left[-\b_{0,1}\right]_{+},\ \ul x_{-1}\right)\in Supp\left(X_{i}\right),
\end{align*}
so that
\begin{align*}
\D\left(\tilde{x};\hat{x},\check{x}\right){}^{'}\b_{0} & =\left(\left|\hat{x}_{1}\right|-\left|\check{x}_{1}\right|\right)\b_{0,1}=1>0,\\
\D\left(\tilde{x};\hat{x},\check{x}\right){}^{'}\b & =\left(\left|\hat{x}_{1}\right|-\left|\check{x}_{1}\right|\right)\b_{1}=-1<0,\\
\D\left(\mathring{x};\hat{x},\check{x}\right){}^{'}\b_{0} & =\left(\left|\hat{x}_{1}-\mathring{x}_{1}\right|-\left|\check{x}_{1}-\mathring{x}_{1}\right|\right)\b_{0,1}=-1<0,\\
\D\left(\mathring{x};\hat{x},\check{x}\right){}^{'}\b & =\left(\left|\hat{x}_{1}-\mathring{x}_{1}\right|-\left|\check{x}_{1}-\mathring{x}_{1}\right|\right)\b_{1}=1>0,
\end{align*}
which again exactly correspond to the inequalities \eqref{eq:ineq2_u}
and \eqref{eq:ineq4_l}. The rest of the proof is exactly the same
as in the proof of Lemma \ref{lem:DNF2}.
\end{proof}

\section{\label{subsec:ID_asym}Asymmetry of Pairwise Observable Characteristics}

So far we have been focusing on the case with symmetric pairwise observable
characteristics, i.e.,
\[
w\left(X_{i},X_{j}\right)\equiv w\left(X_{j},X_{i}\right).
\]
In this section, we briefly discuss how our method can be adapted
to accommodate asymmetric pairwise observable characteristics. 

As in Remark \ref{rem:AymW}, consider the adapted model \eqref{eq:Model_AsymW}:
\begin{equation}
\E\left[\rest{D_{ij}}X_{i},X_{j},A_{i},A_{j}\right]=\phi\left(w\left(X_{i},X_{j}\right)^{'}\b_{0},w\left(X_{j},X_{i}\right)^{'}\b_{0},A_{i},A_{j}\right)\label{eq:asym w}
\end{equation}
where $w$ needs not be symmetric with respect to its two vector arguments
and $\phi:\R^{4}\to\R$ is required to be monotone in all its four
arguments.

The technique of logical differencing still applies in the exactly
same way as before. Specifically, the event $\left\{ \rho_{\ol i}\left(\ol x\right)>\rho_{\ol j}\left(\ol x\right)\right\} $
implies that
\[
\left\{ w\left(X_{\ol i},\ol x\right)^{'}\b_{0}>w\left(X_{\ol j},\ol x\right)^{'}\b_{0}\right\} \text{ OR }\left\{ w\left(\ol x,X_{\ol i}\right)^{'}\b_{0}>w\left(\ol x,X_{\ol j}\right)^{'}\b_{0}\right\} \text{ OR }\left\{ A_{\ol i}>A_{\ol j}\right\} ,
\]
while the event $\left\{ \rho_{\ol i}\left(\ul x\right)<\rho_{\ol j}\left(\ul x\right)\right\} $
implies that
\[
\left\{ w\left(X_{\ol i},\ul x\right)^{'}\b_{0}<w\left(X_{\ol j},\ul x\right)^{'}\b_{0}\right\} \text{ OR }\left\{ w\left(\ul x,X_{\ol i}\right)^{'}\b_{0}<w\left(\ul x,X_{\ol j}\right)^{'}\b_{0}\right\} \text{ OR }\left\{ A_{\ol i}<A_{\ol j}\right\} .
\]
The joint occurrence of $\left\{ \rho_{\ol i}\left(\ol x\right)>\rho_{\ol j}\left(\ol x\right)\right\} $
and $\left\{ \rho_{\ol i}\left(\ul x\right)<\rho_{\ol j}\left(\ul x\right)\right\} $
now implies that
\begin{align}
 & \left\{ w\left(X_{\ol i},\ol x\right)^{'}\b_{0}>w\left(X_{\ol j},\ol x\right)^{'}\b_{0}\right\} \text{ OR }\left\{ w\left(\ol x,X_{\ol i}\right)^{'}\b_{0}>w\left(\ol x,X_{\ol j}\right)^{'}\b_{0}\right\} \nonumber \\
\text{ OR } & \left\{ w\left(X_{\ol i},\ul x\right)^{'}\b_{0}<w\left(X_{\ol j},\ul x\right)^{'}\b_{0}\right\} \text{ OR }\left\{ w\left(\ul x,X_{\ol i}\right)^{'}\b_{0}<w\left(\ul x,X_{\ol j}\right)^{'}\b_{0}\right\} ,\label{eq:ID_rest_asym}
\end{align}
which is in general ``less restrictive'' than the corresponding
identifying restriction in Lemma \ref{lem:DNF1}.

In particular, in the extreme case where $w$ is \emph{antisymmetric}
in the sense of 
\[
w\left(X_{i},X_{j}\right)\equiv-w\left(X_{j},X_{i}\right),
\]
the identifying restriction on the RHS of 
\begin{align*}
 & \left\{ w\left(X_{\ol i},\ol x\right)^{'}\b_{0}>w\left(X_{\ol j},\ol x\right)^{'}\b_{0}\right\} \text{ OR }\left\{ w\left(\ol x,X_{\ol i}\right)^{'}\b_{0}>w\left(\ol x,X_{\ol j}\right)^{'}\b_{0}\right\} 
\end{align*}
becomes 
\[
\left\{ w\left(X_{\ol i},\ol x\right)^{'}\b_{0}\neq w\left(X_{\ol j},\ol x\right)^{'}\b_{0}\right\} ,
\]
which can be generically true and thus becomes (almost) trivial.

Correspondingly, Assumption \ref{assu:DNF-supportX} needs to be strengthened
for point identification:
\begin{assumption*}[\textbf{3a}]
\label{asymmetric W case} There exist a pair of $\ol x,\ul x$,
both of which lie in the support of $Supp\left(X_{i}\right)$, such
that 
\[
Supp\left(w\left(\ol x,X_{i}\right)-w\left(\ul x,X_{i}\right)\right)\cap Supp\left(w\left(X_{i},\ol x\right)-w\left(X_{i},\ul x\right)\right)
\]
contains all directions in $\R^{d}$.
\end{assumption*}
Clearly, the case of antisymmetric $w$ is ruled out by Assumption
3a. Assumption \ref{asymmetric W case} ensures that, for any $\b\neq\b_{0}$,
there exist in-support $x_{i}$ and $x_{j}$ such that 
\begin{align}
 & \left\{ w\left(x_{i},X_{k}\right)^{'}\b_{0}>w\left(x_{j},X_{k}\right)^{'}\b_{0}\right\} \text{ AND }\left\{ w\left(x_{i},X_{l}\right)^{'}\b_{0}<w\left(x_{j},X_{l}\right)^{'}\b_{0}\right\} \nonumber \\
\text{AND} & \left\{ w\left(X_{k},x_{i}\right)^{'}\b_{0}>w\left(X_{k},x_{j}\right)^{'}\b_{0}\right\} \text{ AND }\left\{ w\left(X_{l},x_{i}\right)^{'}\b_{0}<w\left(X_{l},x_{j}\right)^{'}\b_{0}\right\} \label{eq:one direct asym}
\end{align}
and
\begin{align}
 & \left\{ w\left(x_{i},X_{k}\right)^{'}\b\leq w\left(x_{j},X_{k}\right)^{'}\b\right\} \text{ AND }\left\{ w\left(x_{i},X_{l}\right)^{'}\b\geq w\left(x_{j},X_{l}\right)^{'}\b\right\} \nonumber \\
\text{AND} & \left\{ w\left(X_{k},x_{i}\right)^{'}\b\leq w\left(X_{k},x_{j}\right)^{'}\b\right\} \text{ AND }\left\{ w\left(X_{l},x_{i}\right)^{'}\b\geq w\left(X_{l},x_{j}\right)^{'}\b\right\} \label{eq:other direct asym w}
\end{align}
occur simultaneously with strictly positive probability. \eqref{eq:one direct asym}
and \eqref{eq:other direct asym w} are sufficient for $\left\{ \rho_{\ol i}\left(\ol x\right)>\rho_{\ol j}\left(\ol x\right)\right\} $
and $\left\{ \rho_{\ol i}\left(\ul x\right)<\rho_{\ol j}\left(\ul x\right)\right\} $
to occur simultaneously under the maintained assumption on the support
of $A_{i}$. It thus can guarantee point identification of $\b_{0}$.

The estimator can be correspondingly adapted in an obvious manner.

\section{\label{subsec:Sim_Robustness}Additional Simulation Results}

\subsection{Results Varying $N$ and $d$}

In this section, we vary the number of individuals $N$ and $\b_{0}$'s
dimension $d$ to examine how robust our method is against these variations.
We investigate the performance when $N=50,100,200$ and $d=3,4$,
respectively. We maintain the symmetry in $w\left(\cd,\cd\right)$
and other distributional assumptions as in baseline setup. $M$, the
number of $\left(i,j\right)$ pairs used to evaluate objective function,
is set to be 1,000 in all simulations. Note that one could make $M$
larger for larger $N$ to better capture the more information available
from the increase in $N$. In this sense, our results are conservative
below. Results are summarized in Table \ref{tab:C5_Sim_varyND}.
\begin{center}
\begin{table}[h]
\caption{\label{tab:C5_Sim_varyND}Results Varying $N$ and $d$}

\bigskip{}

\noindent \centering{}{\footnotesize{}}%
\begin{tabular}{cccc|cccc}
\hline 
$\phantom{\frac{\frac{1}{1}}{\frac{1}{1}}}$$d=3$ & $\text{rMSE}$ & $\text{MND}$ & MMAD & $\phantom{\frac{\frac{1}{1}}{\frac{1}{1}}}$$d=4$ & $\text{rMSE}$ & $\text{MND}$ & MMAD\tabularnewline
\hline 
$\phantom{\frac{\frac{1}{1}}{\frac{1}{1}}}$$N=50$ & 0.0839 & 0.0724 & 0.0051 & $\phantom{\frac{\frac{1}{1}}{\frac{1}{1}}}$$N=50$ & 0.1119 & 0.1030 & 0.0091\tabularnewline
$\phantom{\frac{\frac{1}{1}}{\frac{1}{1}}}$$N=100$ & 0.0488 & 0.0417 & 0.0053 & $\phantom{\frac{\frac{1}{1}}{\frac{1}{1}}}$$N=100$ & 0.0692 & 0.0647 & 0.0038\tabularnewline
$\phantom{\frac{\frac{1}{1}}{\frac{1}{1}}}$$N=200$ & 0.0334 & 0.029 & 0.0043 & $\phantom{\frac{\frac{1}{1}}{\frac{1}{1}}}$$N=200$ & 0.0543 & 0.0523 & 0.0038\tabularnewline
\hline 
\end{tabular}{\footnotesize\par}
\end{table}
\par\end{center}

The left part of Table \ref{tab:C5_Sim_varyND} shows the performance
of our estimator when $N$ changes and $d$ is fixed at 3. When $N$
increases, rMSE, MND and sum of absolute bias all show moderate decline
in magnitude, indicating the performance is improving. Similar pattern
is also observed for $d=4$. This is intuitive because with more individuals
in the sample, one can achieve more accurate estimation of $\rho_{i}\left(\cd\right)$
and calculation of $\widehat{Q}\left(\beta\right)$. Moreover, we
can see even with a relatively small sample size of $N=50$, the rMSE
is 0.0839 when $d=3$ and 0.1119 when $d=4$, showing that our method
is informative and accurate. When $N=200$, the performance is very
good, with rMSE being as small as 0.0334 and 0.0543 for $d=3$ and
$d=4$, respectively. When we fix $N$ and compare between $d=3$
and $d=4$, it is clear the increase in $d$ adversely affects the
performance of our estimator, with rMSE and MND increasing for each
$N$. Overall, Table \ref{tab:C5_Sim_varyND} provides evidence that
our method is able to estimate $\beta_{0}$ accurately even with a
small sample size.

\subsection{Results Varying $corr$}

Correlation between observable characteristics $X$ and unobservable
fixed effect $A$ is important in network formation models. We show
how our estimator performs when the correlation between $X$ and $A$
varies. Recall that we constructed $A_{i}$ as 
\begin{equation}
A_{i}=corr\times X_{i,1}+\left(1-corr\right)\times\xi_{i}.\label{eq:sim_A}
\end{equation}
We maintain the baseline DGP for $\left(D,X,w,A,\epsilon,\xi,\b_{0}\right)$
as well as $\left(N,M,d\right)$ as in Section \ref{subsec:Setup-of-Simulation}
and vary $corr$ from 0.20 to 0.90. Results are summarized in Table
\ref{tab:C5_Sim_vary_corr}.
\begin{center}
\begin{table}[h]
\caption{\label{tab:C5_Sim_vary_corr}Results Varying $corr$}

\bigskip{}

\noindent \centering{}{\footnotesize{}}%
\begin{tabular}{cccc}
\hline 
$\phantom{\frac{\frac{1}{1}}{\frac{1}{1}}}$$corr$ & rMSE & MND & MMAD\tabularnewline
\hline 
$\phantom{\frac{\frac{1}{1}}{\frac{1}{1}}}$0.20 & 0.0488 & 0.0417 & 0.0053\tabularnewline
$\phantom{\frac{\frac{1}{1}}{\frac{1}{1}}}$0.50 & 0.0489 & 0.0435 & 0.0186\tabularnewline
$\phantom{\frac{\frac{1}{1}}{\frac{1}{1}}}$0.75 & 0.0763 & 0.0690 & 0.0506\tabularnewline
$\phantom{\frac{\frac{1}{1}}{\frac{1}{1}}}$0.90 & 0.1010 & 0.0951 & 0.0743\tabularnewline
\hline 
\end{tabular}{\footnotesize\par}
\end{table}
\par\end{center}

It can be seen from Table \ref{tab:C5_Sim_vary_corr} that even though
increase in $corr$ adversely affects the performance of our estimator,
the magnitude of the impact is relatively small. For example, rMSE
only increases from 0.0488 to 0.1010 when $corr$ increase dramatically
from 0.20 to 0.90. Similar pattern is also observed using other performance
metrics. Therefore, our estimator is robust against correlation between
$X$ and $A$.

\subsection{Results under Asymmetric Pairwise Observable Characteristics}

Following the theoretical analysis in Section \ref{subsec:ID_asym},
we investigate how our method works when $w\left(X_{i},X_{j}\right)$
is asymmetric in $\left(X_{i},X_{j}\right)$, i.e., $W_{ij,h}:=w_{h}\left(X_{i},X_{j}\right)\neq w_{h}\left(X_{j},X_{i}\right)=:W_{ji,h}$
for at least one coordinate $h\in\left\{ 1,..,d\right\} $. To do
so, we let $W_{ij,d}=\abs{2X_{i,d}-X_{j,d}}\times\left(2/3\right)$
\footnote{The reason for multiplying 2/3 is to make the size of $W_{ij,d}$
similar to other coordinates of $W_{ij}$.} for the last dimension $d$, while setting $W_{ij,h}=\abs{X_{i,h}-X_{j,h}}$
for all other coordinates $h=1,..,d-1$ such that $W_{ij}\neq W_{ji}$
unless $\abs{X_{i,d}}=\abs{X_{j,d}}$, which is a probability zero
event. We maintain the DGP for $\left(X,D,A,\epsilon,\xi,corr,\b_{0}\right)$
as in Section \ref{subsec:Setup-of-Simulation} and fix the number
of $\left(i,j\right)$ pairs $M$ at 1,000 for evaluation of $\widehat{Q}\left(\beta\right)$.
Finally, we vary $N$ and $D$ under the asymmetric $W_{ij}$ setting.
\begin{center}
\begin{table}[h]
\caption{\label{tab:C5_Sim_asymW}Results under Asymmetry}

\bigskip{}

\noindent \centering{}{\footnotesize{}}%
\begin{tabular}{cccc|cccc}
\hline 
$\phantom{\frac{\frac{1}{1}}{\frac{1}{1}}}$$d=3$ & $\text{rMSE}$ & $\text{MND}$ & MMAD & $\phantom{\frac{\frac{1}{1}}{\frac{1}{1}}}$$d=4$ & $\text{rMSE}$ & $\text{MND}$ & MMAD\tabularnewline
\hline 
$\phantom{\frac{\frac{1}{1}}{\frac{1}{1}}}$$N=50$ & 0.1498 & 0.1403 & 0.0936 & $\phantom{\frac{\frac{1}{1}}{\frac{1}{1}}}$$N=50$ & 0.2225 & 0.2124 & 0.1521\tabularnewline
$\phantom{\frac{\frac{1}{1}}{\frac{1}{1}}}$$N=100$ & 0.1096 & 0.1028 & 0.0741 & $\phantom{\frac{\frac{1}{1}}{\frac{1}{1}}}$$N=100$ & 0.1751 & 0.1695 & 0.1301\tabularnewline
$\phantom{\frac{\frac{1}{1}}{\frac{1}{1}}}$$N=200$ & 0.0943 & 0.0893 & 0.0672 & $\phantom{\frac{\frac{1}{1}}{\frac{1}{1}}}$$N=200$ & 0.1595 & 0.1555 & 0.1222\tabularnewline
\hline 
\end{tabular}{\footnotesize\par}
\end{table}
\par\end{center}

Table \ref{tab:C5_Sim_asymW} shows that our method performs reasonably
well when $W_{ij}$ is asymmetric. First, when the number of individuals
$N$ increases, the overall performance is improved, with rMSE decreasing
from 0.1498 to 0.0943 for $d=3$ when $N$ increases from 50 to 200
(similar pattern for $d=4$). This result is caused by the more information
available in the sample and echos what we have seen for the symmetric
$W_{ij}$ case. When $d$ increases from 3 to 4, the performance becomes
worse, with, for instance, rMSE increasing from 0.0943 to 0.1595 for
$N=200$. It shows that more data (information) is required for accurate
estimation when the dimension of $\beta_{0}$ is larger. Second, when
compared with the symmetric $W_{ij}$ case, the overall performance
under asymmetric $W_{ij}$ is generally not as good, with rMSE being
0.1498 for asymmetric $W_{ij}$ versus 0.0839 for symmetric $W_{ij}$
when $N=50$ and $d=3$. We have shown in Appendix \ref{subsec:ID_asym}
the identifying power of the objective function is in general ``less
restrictive'' than the corresponding identifying restriction in Lemma
\ref{lem:DNF1}. Therefore, one would expect larger bias than symmetric
$W_{ij}$ case, which is consistent with what we see in Table \ref{tab:C5_Sim_asymW}.
Based on results in Table \ref{tab:C5_Sim_asymW}, when $W_{ij}$
is asymmetric and computational power is sufficient, we recommend
increasing $M$, the number of $\left(i,j\right)$ pairs to evaluate
$\widehat{Q}\left(\b\right)$, to improve performance.
\end{document}